\documentclass[12pt]{article}
\usepackage[latin1]{inputenc}
\usepackage{geometry}
\usepackage{amsfonts}
\usepackage{amsmath}
\usepackage{amssymb}
\usepackage{amsthm}
\usepackage{graphics}
\usepackage{algorithm}
\usepackage{algpseudocode}

\usepackage{tikz}

\newcommand{\cut}{\cap}
\newcommand{\uni}{\cup}
\newcommand{\Buni}{\bigcup}
\newcommand{\duni}{\dot{\cup}}
\newcommand{\EG}[3]{{#1}[{#2},{#3}]}  
\newcommand{\eG}[2]{\EG{G}{{#1}}{{#2}}}  
\newcommand{\SUU}[2]{{#1}^{{#2}\times{#2}}}   
\newcommand{\sUU}[1]{\SUU{s}{#1}}
\newcommand{\SUV}[3]{{#1}^{{#2}\times{#3}}}    
\newcommand{\sUV}[2]{\SUV{s}{#1}{#2}}
\newcommand{\Sc}{\operatorname{scen}}   
\newcommand{\se}{\rtimes}  

\newcommand{\prc}[1]{\ensuremath{\mathsf{#1}}}

\newcommand{\NP}{\prc{NP}}

\newcommand{\SP}{\prc{\#P}}

\newcommand{\defexpr}[1]{{\em{#1}}}
\newcommand{\poly}{\mathsf{poly}}

\newcommand{\rk}{\operatorname{rk}}
\newcommand{\n}{\operatorname{n}}

\newcommand{\VV}{\widetilde{V}}
\newcommand{\UU}{\widetilde{U}}
\newcommand{\subs}{\subseteq}

\newcommand{\lc}[2]{{#1}^{{#2}}}   
\newcommand{\ep}[2]{{#1}^{{#2}}}   
\newcommand{\slt}{\nabla} 
\newcommand{\st}[3]{s_{\operatorname{join}}({#1}, {#2},{#3})} 
\newcommand{\sm}[3]{s_{\operatorname{forget}}({#1}, {#2}, {#3})}  
\renewcommand{\sp}[3]{s_{\operatorname{introduce}}({#1}, {#2}, {#3})}  
\newcommand{\si}[2]{s_{\operatorname{ignore}}({#1}, {#2})}  

\newcommand{\s}{\tilde{s}}
\renewcommand{\S}{\mathcal{S}}
\newcommand{\rmi}{\mathrm{i}}
\newcommand{\rmf}{\mathrm{f}}

\newcommand{\dr}[3]{\Delta r_{\operatorname{forget}}({#1}, {#2},{#3})} 
\newcommand{\dn}[3]{\Delta n_{\operatorname{forget}}({#1}, {#2},{#3})}
\newcommand{\zo}{\{0,1\}}  
\newcommand{\zot}{\{0,1,2\}}  
\newcommand{\tw}{\operatorname{tw}}
\renewcommand{\AA}{\widetilde{A}}
\newcommand{\CC}{\mathcal{C}}
\newcommand{\M}{\mathcal{M}}
\newcommand{\D}{\mathcal{D}}

\newcommand{\trunc}[2]{{#1} | {#2}}

\newtheorem{defi}{Definition}[section]
\newtheorem{cor}[defi]{Corollary}
\newtheorem{lem}[defi]{Lemma}
\newtheorem{thm}[defi]{Theorem}

\algnewcommand\Input{\item[\algorithmicinput]}%
\algnewcommand\algorithmicinput{\textbf{Input:}}

\begin{document}

\title{Fast Evaluation of Interlace Polynomials on Graphs of Bounded Treewidth}

\author{Markus Bl\"aser \and Christian Hoffmann}

\maketitle

\begin{abstract}
  We consider the multivariate interlace polynomial introduced by
  Courcelle (2008), which generalizes several interlace polynomials
  defined by Arratia, Bollobás, and Sorkin (2004) and by Aigner and
  van der Holst (2004). We present an algorithm to evaluate the
  multivariate interlace polynomial of a graph with $n$ vertices given
  a tree decomposition of the graph of width $k$. The best previously
  known result (Courcelle 2008) employs a general logical framework
  and leads to an algorithm with running time $f(k)\cdot n$, where
  $f(k)$ is doubly exponential in~$k$. Analyzing the $GF(2)$-rank of
  adjacency matrices in the context of tree decompositions, we give a
  faster and more direct algorithm. Our algorithm uses
  $2^{3k^2+O(k)}\cdot n$ arithmetic operations and can be efficiently
  implemented in parallel.
\end{abstract}

\section{Introduction}

Inspired by some counting problem arising from DNA sequencing
\cite{interlace_dna}, Arratia, Bollobás, and Sorkin defined a graph
polynomial which they called interlace polynomial
\cite{interlace_polynomial}.  It turned out that the interlace
polynomial is related \cite[Theorem 24]{interlace_polynomial} to the
Martin polynomial, which counts the number of edge partitions of a
graph into circuits.  This polynomial has been defined in Martin's
thesis from 1977 \cite{martin_thesis} and generalized by Las Vergnas
\cite{Las_Vergnas_Martin_Polynomial}.  Further work on the Martin
polynomial has been pursued \cite{Las_Vergnas_Eul,
  Las_Vergnas_Tutte33, Jaeger_Tutte_cycles_plane, emg_martin_new,
  emg_martin_misc, bb_eval_circuit_partition}, including a
generalization to isotropic systems \cite{bouchet_isotropic,
  bouchet_graphic_repr, bouchet_tutte_martin, bouchet_et_al}. In
particular, the Tutte polynomial of a planar graph and the Martin
polynomial of its medial graph are related. This implies a connection
between the Tutte polynomial and the interlace polynomial (see
Ellis-Monaghan and Sarmiento \cite{ems_distance_hereditary} for an
explanation).

One way to define the interlace polynomial is by a recursion that uses
a graph operation. Arratia et al.\ used a pivot operation for edges
\cite{interlace_polynomial}. This operation is a composition of local
complementations to neighbor vertices (see Aigner and van der Holst
\cite{aigner_holst}, where the operations are called switch
operations). The orbits of graphs under local complementation are
related to error-correcting codes and quantum states, and so is the
interlace polynomial as well \cite{danielsen_parker}.

The interlace polynomial can also be defined by a closed expression
using the $GF(2)$-rank of adjacency matrices \cite{aigner_holst,
  Bouchet05, ems_interlace_isotropic}. This linear algebra approach
has been used in several generalizations of the interlace polynomial.
In this paper, we consider the multivariate interlace polynomial
$C(G)$ defined by Courcelle \cite{courcelle_interlace_final} (see
Definition~\ref{def:interlace_polynom} below) as it subsumes the
two-variable interlace polynomial of Arratia, Bollobás, and Sorkin
\cite{arratia_two_var_interl} and the weighted versions of Traldi
\cite{traldi_weighted_interlace}, as well as the interlace polynomials
defined by Aigner and van der Holst \cite{aigner_holst}.  Furthermore,
the interlace polynomials $Q(x,y)$ and $Q_n^{HN}$, which have emerged
from a spectral view on the interlace polynomials \cite{riera_parker},
are also special cases of Courcelle's multivariate interlace
polynomial.

\subsection{Results and Related Work}

Our aim is to present an algorithm that, given a graph $G=(V,E)$ and
an evaluation point, i.e.\ a tuple of numbers $((x_a)_{a\in V},
(y_a)_{a\in V}, u, v)$, evaluates the multivariate interlace
polynomial $C(G)$ at $((x_a)_{a\in V}, (y_a)_{a\in V}, u, v)$.
Whereas this is a $\SP$-hard problem in general \cite{interlace_hard},
it is fixed parameter tractable with cliquewidth as parameter
\cite[Theorem 23, Corollary 33]{courcelle_interlace_final}. This is a
consequence of the fact that the interlace polynomial is monadic
second order logic definable ($MS_1$ definable as defined by
Courcelle, Makowsky, and Rotics
\cite{DBLP:journals/dam/CourcelleMR01}; see also Courcelle
\cite[Section 5]{courcelle_interlace_final}).\footnote{Note the
  following crucial difference with respect to monadic second order
  logic definability: $MS_1$ definable evaluation problems are fixed
  parameter tractable with \emph{cliquewidth} as parameter
  \cite[Theorem 31]{DBLP:journals/dam/CourcelleMR01}. $MS_1$ is a
  logic that allows one-sorted structures, the universe of which
  consists of the vertices of the graph. Set variables range over
  vertex subsets. On the contrary, $MS_2$ is a logic that allows
  two-sorted structures, the universe of which consists of the
  vertices \emph{and edges} of the graph. Set variables range over
  vertex subsets or edge subsets, which, for instance, enables the
  definition of the Tutte polynomial in $MS_2$. $MS_2$ definable
  evaluation problems are known to be fixed parameter tractable with
  \emph{treewidth} as parameter \cite[Theorem
    32]{DBLP:journals/dam/CourcelleMR01}. We can not expect that this
  generalizes to cliquewidth, see Fomin, Golovach, Lokshtanov, and
  Saurabh \cite{fomin_cliquewidth}.} Such graph polynomials can be
evaluated in time $f(k)\cdot n$, where $n$ is the number of vertices
of the graph and $k$ is the cliquewidth. The function $f(k)$ can be
very large and is not explicitly stated in most cases.  In general, it
grows as fast as a tower of exponentials the height of which is
proportional to the number of quantifier alternations in the formula
describing the graph polynomial \cite[Page
  34]{courcelle_interlace_final}. In the case of the interlace
polynomial, this formula involves two quantifier alternations
\cite[Lemma 24]{courcelle_interlace_final},
\cite{DBLP:journals/jct/CourcelleO07}. If a graph has tree width $k$,
its cliquewidth is bounded by $2^{k+1}$
\cite{DBLP:journals/dam/CourcelleO00}. Thus, the machinery of monadic
second order logic implies the existence of an algorithm that
evaluates the interlace polynomial of an $n$-vertex graph in time
$f(k)\cdot n$, where $k$ is the tree width of the graph and $f(k)$ is
at least doubly exponential in $k$. (In particular, the interlace
polynomial of graphs of treewidth~1, that is, of trees, can be
evaluated in polynomial time, which also has been observed by Traldi
\cite{traldi_weighted_interlace}.)

The monadic second order logic approach is very general and can be
applied not only to the interlace polynomial but to a much wider class
of graph polynomials \cite{DBLP:journals/dam/CourcelleMR01}. However,
it does not consider characteristic properties of the actual graph
polynomial. In this paper, we restrict ourselves to the interlace
polynomial so as to exploit its specific properties and to gain a more
efficient algorithm (Algorithm~\ref{alg:main}).  Our algorithm
performs $2^{3k^2+O(k)}n$ arithmetic operations to evaluate
Courcelle's multivariate interlace polynomial (and thus any other
version of the interlace polynomial mentioned above) on an $n$-vertex
graph given a tree decomposition of width $k$
(Theorem~\ref{thm:algo}). The algorithm can be implemented in parallel
using depth polylogarithmic in $n$ (Section~\ref{ssec:parallel}).
Apart from evaluating the interlace polynomial, our approach can also
be used to compute coefficients of the interlace polynomial, for
example so called $d$-truncations \cite[Section
  5]{courcelle_interlace_final} (Section~\ref{ssec:coefficients}). Our
approach is not via logic but via the $GF(2)$-rank of adjacency
matrices, which is specific to the interlace polynomial.

\subsection{Obstacles}

It has been noticed that the Tutte polynomial and the interlace
polynomial are similar in some respect \cite{arratia_two_var_interl}:
Both can be defined by a recursion using a graph operation, both can
be defined as closed sums over edge/vertex subsets involving some kind
of rank. These similarities suggest that evaluating the interlace
polynomial using tree decompositions might work completely analogously
to the respective approaches for the Tutte polynomial \cite{Andrzejak,
  Noble}. This is not the case because of the following problems.

Andrzejak's algorithm \cite{Andrzejak} to evaluate the Tutte
polynomial uses the deletion-contraction recursion for the Tutte
polynomial (via Negami's splitting formula \cite{Negami}). Deletion
and contraction of an edge has the nice property that it is compliant
with tree decompositions: If we are given the tree decomposition of a
graph and we delete (or contract) an edge, the original tree
decomposition (or, in the case of edge contraction, a simple
modification of it) is a tree decomposition of the modified graph. For
the interlace polynomial, on the other hand, the respective graph
operation is not compliant with tree decompositions: If we perform the
pivot operation from Arratia et al.\ \cite{interlace_polynomial} on a
graph, it is not clear how to obtain a tree decomposition of the
modified graph.  In particular, a single pivot operation can turn a
tree (treewidth 1) into a cycle (treewidth 2), see
Fig.~\ref{fig:pivot_increases_treewidth}.

\begin{figure}
  \centering
  \includegraphics{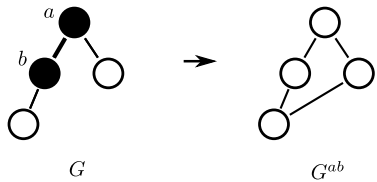}  
  \caption{Edge pivoting, a central graph operation for the interlace
    polynomial, increases treewidth.}
  \label{fig:pivot_increases_treewidth}
\end{figure}

Another problem is that in the Tutte case the recursion formula
naturally generalizes from the simplest versions (chromatic
polynomial) to the most general ones (it is the defining recursion of
the Bollobás-Riordan graph invariant \cite{bollobas_riordan};
cf.\ also the recurrence relation of the polynomial of Averbouch,
Godlin and Makowsky, which generalizes the Tutte polynomial and the
matching polynomial \cite{DBLP:conf/wg/AverbouchGM08}). The interlace
polynomial, in contrast, needs more and more complicated recursions
when generalizing the vertex-nullity interlace polynomial to the
multivariate interlace polynomial\footnote{But note that Traldi
  reduced a three-term recursion to a two-term recursion
  \cite[Corollary~2.4]{traldi_weighted_interlace}.} (see Courcelle
\cite[Proposition 12]{courcelle_interlace_final}).

When we consider Noble's algorithm \cite{Noble} and concentrate on the
definition of the Tutte / interlace polynomial by sums involving
ranks, another problem emerges. In the Tutte case, the rank is an easy
to understand graph theoretic value, namely the number of vertices
minus the number of connected components. Noble observes that, if a
graph is extended by a set of vertices and some edges between the old
and the new vertices, the set of all partitions of the new vertices
captures all possible types of ``behavior'' of the rank (i.e.\ number
of connected components) when the new vertices and some or all of the
new edges are added. -- For the interlace polynomial on the other
hand, the rank used in the definition is the rank over $GF(2)$ of the
adjacency matrix. Even though there exists a graph theoretic
interpretation of this rank \cite{traldi_binary_nullity}, it is
substantially more involved. Furthermore, an appropriate tool to
capture the ``rank behavior'' when extending a graph (such as vertex
partitions in the case of the Tutte polynomial) seems to be
missing. The main contribution of this work is to devise such a tool
and to prove that it works well with tree decompositions.

\subsection{Outline}

We compute the interlace polynomial by dynamic programming on the tree
decomposition of a graph. To this end, we analyze the behavior of the
$GF(2)$-rank of the adjacency matrix of a graph when the graph is
extended by a fixed number of vertices and edges between these new
vertices and the existing ones.

Section~\ref{sec:prelim} contains the definition of Courcelle's
multivariate interlace polynomial, which we will consider in this
work. We will also fix our notation for tree decompositions there. In
Section~\ref{sec:ideas} we present our approach in detail. This
includes the motivation and definition of two central terms: extended
graphs and scenarios.  A scenario captures the behavior of the rank of
an adjacency matrix when adding vertices. To define this precisely, we
introduce symmetric Gaussian elimination in
Section~\ref{sec:symgauss}. In Section~\ref{sec:scenarios}, we collect
properties of scenarios which enable us to use scenarios with tree
decompositions. In Section~\ref{sec:algo}, we describe and analyze our
algorithm, which evaluates the interlace polynomial by splitting it
into parts according to scenarios. In Section~\ref{sec:variants} we
discuss how our algorithm can be parallelized and used to compute
(some of the) coefficients of the interlace polynomial. Finally, in
Section~\ref{sec:open}, we mention directions for further research.

\section{Preliminaries}
\label{sec:prelim}

We consider undirected graphs without multiple edges but with self
loops allowed. Let $G=(V,E)$ be such a graph and $A\subs V$. By $G[A]$
we denote the subgraph of $G$ induced by $A$, i.e.\ $(A, \{e\ |\ e\in
E, e \subs A \})$. $G\slt A$ denotes the graph $G$ with self loops in
$A$ toggled, i.e.\ the graph obtained from $G$ by performing the
following operation for each vertex $a\in A$: if $a$ has a self loop,
remove it; if $a$ does not have a self loop, add one.

The adjacency matrix of $G$ is a symmetric square matrix with entries
from $\zo$. As the matrices that we will consider are adjacency
matrices of graphs, we use vertices as column/row indices.  Thus, the
adjacency matrix of $G$ is a $V\times V$ matrix $M=(m_{uv})$ over
$\zo$ with $m_{uv}=1$ iff $uv\in E$. Furthermore, we will refer to
entries and submatrices by specifying first the rows and then the
columns: the $(u,v)$-entry of $M=(m_{uv})$ is $m_{uv}$, the $A\times
B$ submatrix of $M$ is the submatrix of the entries of $M$ with row
index in $A$ and column index in $B$. All matrix ranks will be ranks
over the field with two elements, $\zo=GF(2)$, i.e.\ $+$ is XOR and
$\cdot$ is AND. Slightly abusing notation we write $\rk(G)$ for the
rank of the adjacency matrix of the graph $G$. The nullity (or
co-rank) of an $n\times n$ matrix $M$ is $\n(M)=n-\rk(M)$. If $G$ is a
graph, we write $\n(G)$ for the nullity of the adjacency matrix of
$G$.

\emph{Graph polynomials} are, from a formal perspective, mappings of
graphs to polynomials that respect graph isomorphism. We will consider
a \emph{multivariate} graph polynomial, the multivariate interlace
polynomial. To define such a polynomial, one has to distinguish
``ordinary'' indeterminates from \emph{$G$-indexed indeterminates}.
For instance, $x$ being a $G$-indexed indeterminate means that for
each vertex $a$ of $G$ there is a different indeterminate $x_a$. If
$A\subs V$, we write $x_A$ for $\prod_{a\in A} x_a$. 


\begin{defi}[Courcelle \cite{courcelle_interlace_final}]
\label{def:interlace_polynom}
  Let $G=(V,E)$ be an undirected graph. The multivariate interlace
  polynomial is defined as
  \begin{equation}
    \label{eq:interlace_polynom}
    C(G)=\sum_{\substack{A,B\subs V\\ A\cut B =\emptyset}} x_A y_B u^{\rk((G\slt B)[A\uni B])} v^{\n((G\slt
    B)[A\uni B])}, 
  \end{equation}
  where
  $u,v$ are called ordinary indeterminates and $x,y$ $G$-indexed
  indeterminates.
\end{defi}

\subsection{Tree Decompositions}
\label{ssec:treedecomp}

We borrow most of our notation from Bodlaender and Koster
\cite{DBLP:journals/cj/BodlaenderK08}. A \defexpr{tree decomposition}
of a graph $G=(V,E)$ is a pair $(\{X_i\ |\ i\in I\}, T=(I,F))$ where
$T$ is a tree and each node $i\in I$ has a subset of vertices
$X_i\subs V$ associated to it, called the bag of $i$, such that the
following holds:
\begin{enumerate}
\item Each vertex belongs to at least one bag, that is $\Buni_{i\in I}
  X_i = V$.
\item Each edge is represented by at least one bag, i.e.\ for all
  $e=vw\in E$ there is an $i\in I$ with $v,w\in X_i$.
\item For all vertices $v\in V$, the set of nodes $\{i \in I\ |\ v\in
  X_i\}$ induces a connected subgraph of $T$.
\end{enumerate}
The width of a tree decomposition $(\{X_i\},T)$ is $\max\{|X_i|\ |\
i\in I\}-1$. The treewidth of a graph $G$, $\tw(G)$, is the minimum
width over all tree decompositions of $G$.

Computing the treewidth of a graph is $\NP$-complete. But given a
graph with $n$ vertices, we can compute a tree decomposition of width
$k$ (or detect that none exists) using Bodlaender's algorithm in time
$2^{O(k^3)}n$ \cite{bodlaender} (cf.\ also Downey and Fellows
\cite[Section 6.3]{downey_fellows}).

To evaluate the interlace polynomial we will use \emph{nice} tree
decompositions. Note that our definition slightly deviates from the
usual one\footnote{Usually, there is no special restriction on the bag
  size of the root node, and the leaf nodes contain exactly \emph{one}
  vertex.}.  This has no substantial influence on the running time of
the algorithms discussed in this work, but it simplifies the
presentation. In a nice tree decomposition $(\{X_i\}, T)$, $T$ is a
rooted tree with $|X_r|=0$ for the root $r$ of $T$, and each node $i$
of $T$ is of one of the following types:
\begin{itemize}
\item Leaf: node $i$ is a leaf of $T$ and $|X_i|=0$.
\item Join: node $i$ has exactly two children $j_1$ and $j_2$, and
  $X_i=X_{j_1}=X_{j_2}$.
\item Introduce: node $i$ has exactly one child $j$, and there is a
  vertex $a\in V\setminus X_j$ with $X_i = X_j\uni \{a\}$.
\item Forget: node $i$ has exactly one child $j$, and there is a
  vertex $a\in V\setminus X_i$ with $X_j = X_i\uni \{a\}$.
\end{itemize}

A tree decomposition of width $k$ with $n$ nodes can be converted into
a nice tree decomposition of width $k$ with $O(n)$ nodes in time
$O(n)\cdot \poly(k)$ \cite[Lemma 13.1.2, 13.1.3]{Kloks}.

For a graph $G$ with a nice tree decomposition $(\{X_i\}, T)$, we
define
\[V_i=\Bigl( \Buni \{X_j\ |\ \text{$j$ is in the subtree of $T$ with
  root $i$} \} \Bigr) \setminus X_i\quad \text{and}\quad G_i =
G[V_i].\] We can think of $G_i$ as the subgraph of $G$ induced by all
vertices that have already been forgotten below node $i$.

\section{Idea}
\label{sec:ideas}

We will now sketch our idea how to evaluate the interlace polynomial.
Our approach is dynamic programming similar to the work of Noble
\cite{Noble}.  Let $G$ be a graph for which we want to evaluate the
interlace polynomial and $(\{X_i\}, T)$ a nice tree decomposition of
$G$.  For each node $i$ of the tree decomposition, we have defined the
graph $G_i$ that consists of all vertices in the bags below $i$ that
are not in $X_i$.  We will compute ``parts'' of the interlace
polynomial of $G_i$. These parts are essentially defined by the answer
to the following question: How does the rank of the adjacency matrix
of some subgraph of $G_i$ increase when we add (some or all) vertices
of $X_i$? For the leaves these parts are trivial.  Our algorithm
traverses the tree decomposition bottom-up.  We will show how to
compute the parts of an introduce, forget, or join node from the parts
of its child node (children nodes, resp.).  At the root node, there is
only one part left. This part is the interlace polynomial of $G$.

\label{rank_behavior_discussion}
Before we go into details, let us remark that the answer to the above
question (``How does the rank of the adjacency matrix increase when
adding some vertices?'')  depends on the internal structure of the
graph being extended. Consider the graph on the left hand side in
Figure~\ref{fig:rankincr2}. If we extend it by the black vertices, the
rank increases by $2$. But if we use the graph on the left hand side
in Figure~\ref{fig:rankincr4}, the \emph{same extension} causes a rank
increase by $4$.

\begin{figure}
  \centering
  \includegraphics{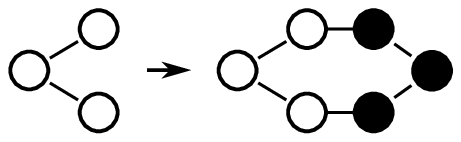}  
  \caption{Interlace polynomial and rank behavior: Rank over $GF(2)$
    of the adjacency matrix increases by 2 (from 2 to 4).}
  \label{fig:rankincr2}
\end{figure}
\begin{figure}
  \centering
  \includegraphics{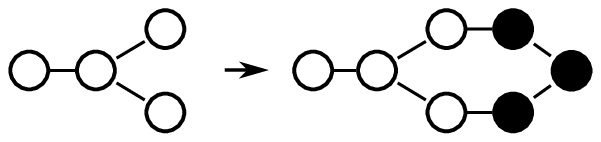}  
  \caption{Interlace polynomial and rank behavior: Rank over $GF(2)$
    of the adjacency matrix increases by 4 (from 2 to 6).}
  \label{fig:rankincr4}
\end{figure}

Let us see how we handle this issue. We start with the following
definition.

\begin{defi}[Extended graph]
  Let $G=(V,E)$ be some graph, $V', U\subs V$, $V'\cut U = \emptyset$.
  Then we define $\eG{V'}{U}$ to denote $G[V'\uni U]$ and call
  $\eG{V'}{U}$ an \defexpr{extended graph}, the graph obtained by
  \defexpr{extending $G[V']$ by $U$ according to $G$}. We call $U$ the
  \defexpr{extension of $\eG{V'}{U}$}.
\end{defi}

Let us fix an extension $U$. We consider all $V'\subs V(G)$ such that
$G[V']$ may be extended by $U$ according to the input graph $G$. For
any such extended graph we ask: ``How does the rank of $G[V']$
increase when adding some vertices of $U$?''. Our key observation is
that the answer to this question can be given without inspecting the
actual $G$ if we are provided with a compact description (of size
independent of $n=|V(G)|$), which we call the scenario of $\eG{V'}{U}$.

\label{lab:scenario_motivation}
The scenario of $\eG{V'}{U}$ (Definition~\ref{def:extension_scenario})
will be constructed in the following way.  Consider $M$, the adjacency
matrix of $G[V'\uni U]$. Perform symmetric Gaussian elimination on $M$
\emph{using only the vertices in $V'$} (for the details see
Section~\ref{sec:symgauss}).  The resulting matrix $M'$ is symmetric
again and has the same rank as $M$. Furthermore, $M'$ is of a form as
in Figure~\ref{fig:scenario_intuition01}: The $V'\times V'$ submatrix
is a symmetric permutation matrix with some additional zero
columns/rows.  The nonzero entries correspond to edges or self loops
(not of the original graph $G$ but of some modified graph that is
obtained from $G$ in a well-defined way) ``ruling'' over their
respective columns/rows: The edge between $v_1$ and $v_8$ rules over
columns and rows $v_1$ and $v_8$. Here, ``to rule'' means that the
only $1$s in these columns and rows are the $1$s at $(v_1,v_8)$ and
$(v_8, v_1)$.  Similarly, the self loop at vertex $v_5$ rules over
column and row $v_5$. The columns and rows that are ruled by some edge
or self loop in $V'$ are also empty (i.e.\ entirely zero) in the
$U\times V'$ submatrix of $M'$. Some columns/rows are not ruled by any
edge or self loop in $V'$, for instance column/row $v_4$.  This is
because there is neither a self loop at vertex $v_4$ nor does it have
a neighbor in $V'$.  However, $v_4$ may have neighbors in $U$.  Thus,
column $v_4$ of the $U\times V'$ submatrix may be any value from
$\zo^U$, which is indicated by the question marks. Also, the contents
of the $U\times U$ submatrix is not known to us.

\begin{figure}
  \centering
  \includegraphics{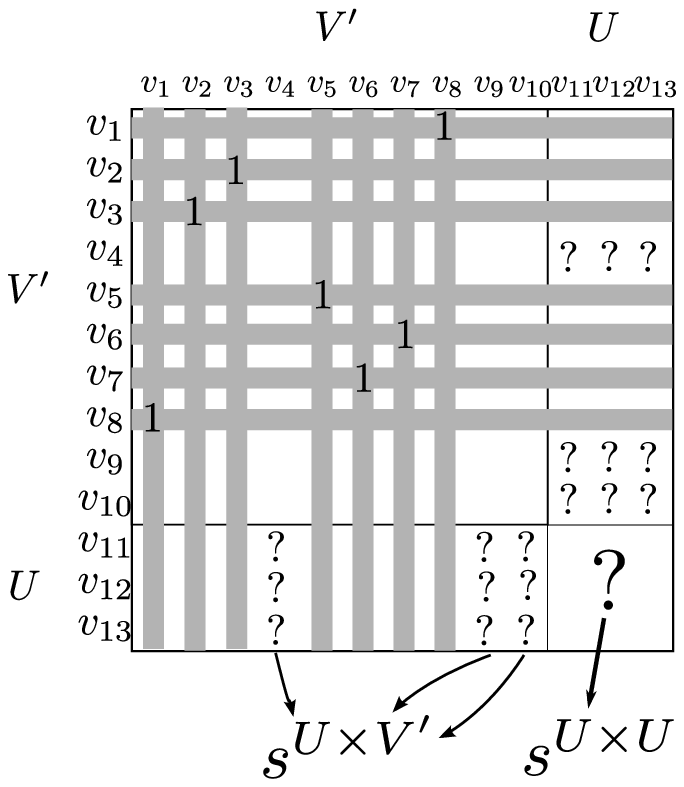}
  \caption{Adjacency matrix of $G[V'\uni U]$ after symmetric Gaussian
    elimination using $V'$. Empty entries are zero.}
  \label{fig:scenario_intuition01}
\end{figure}

Let us choose a basis of the subspace spanned by the nonzero columns
of the $U\times V'$ submatrix and call it $\sUV{U}{V'}$.  Let
$\sUU{U}$ be contents of the $U\times U$ submatrix.  By this
construction, we are able to describe the rank of $M'$ as the rank of
its $V'\times V'$ submatrix plus a value that can be computed solely
from $\sUV{U}{V'}$ and $\sUU{U}$.

This will solve our problem that the rank increase depends on the
internal structure of the graph $G[V']$ being extended: all we need to
know is the scenario $s=(\sUV{U}{V'}, \sUU{U})$ of $\eG{V'}{U}$.  From $s$,
without considering $G[V']$, we can compute in time $\poly(|U|)$ how
the rank of the adjacency matrix of $G[V']$ increases when we add some
vertices from $U$. This motivates the following definition.

\begin{defi}[Scenario]
\label{def:scenario}
Let $U$ be an extension, i.e.\ a finite set of vertices. A
\defexpr{scenario of $U$} is a tuple $s=(\sUV{U}{V'},\sUU{U})$ where
$\sUV{U}{V'}$ is an ordered set of linearly independent vectors spanning a
subspace of $\zo^U$ and $\sUU{U}$ is a symmetric $(U \times U)$-matrix
with entries from $\zo$.  A \defexpr{scenario for $k$ vertices} is a
scenario of some vertex set $U$ with $|U|=k$.
\end{defi}

Let us come back to the evaluation of the interlace polynomial of $G$
using a tree decomposition. Recall that at a node $i$ of the tree
decomposition we want to compute ``parts'' of the interlace polynomial
of $G[V_i]$. Essentially every scenario $s$ of $X_i$ will define such
a part: The interlace polynomial itself is a sum over \emph{all}
induced subgraphs with self loops toggled for some vertices. The part
of the interlace polynomial corresponding to scenario $s$ will be the
respective sum not over all these graphs but only over the ones such
that $s$ is the scenario of $\eG{V_i}{X_i}$. This will lead us to
\eqref{eq:bubble_term} in Section~\ref{sec:blubble_term}. To compute
the parts of a join, forget and introduce node from the parts of its
children nodes (child node, resp.), we will employ
Lemma~\ref{lem:sum_join}, \ref{lem:sum_introduce} and
\ref{lem:sum_forget}. These are based on the fact that scenarios are
compliant with tree decompositions, which we will prove in
Section~\ref{sec:scenarios} (Lemma~\ref{lem:rank_join},
Lemma~\ref{lem:rank_introduce} and Lemma~\ref{lem:rank_forget}). An
example for the overall procedure of the algorithm is depicted in
Figure~\ref{fig:algo_example}.

\begin{figure}[!h]
  \centering
  \includegraphics{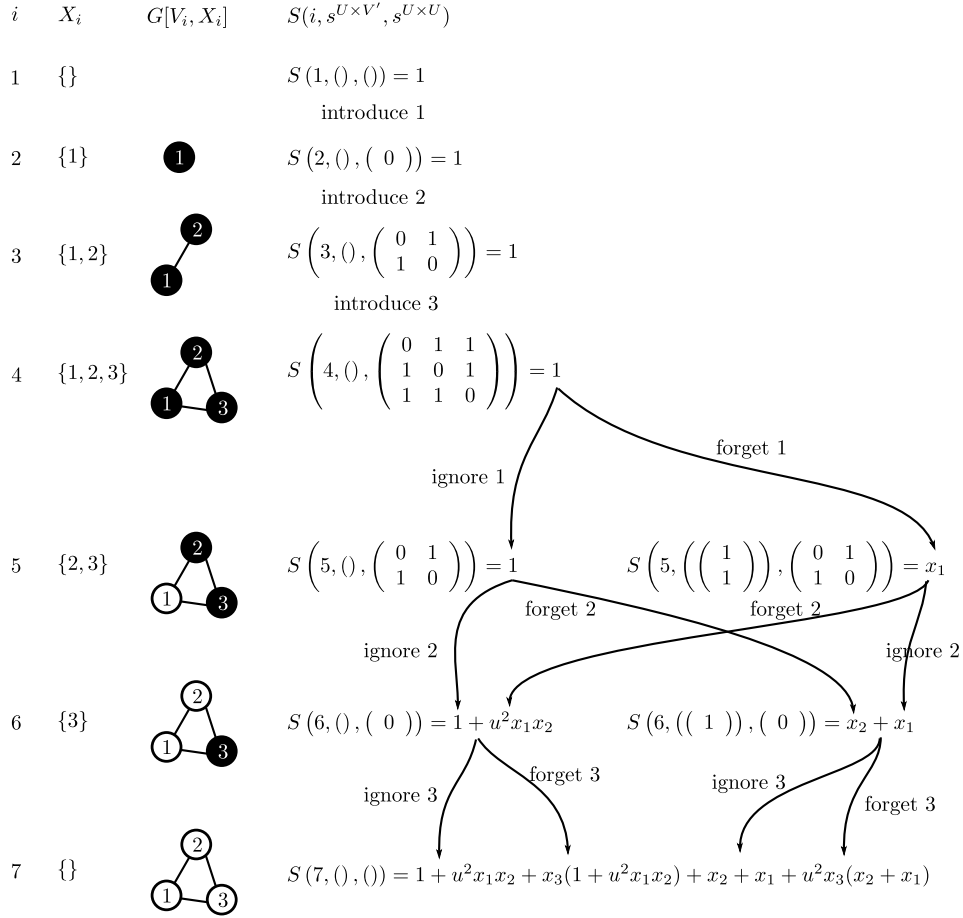}
  \caption{Computation of the interlace polynomial $C(G;y=0,v=1)$ of a
    triangle. In order to simplify the illustration, we ignore
    parameter $D$ in \eqref{eq:bubble_term}, which handles the ``self
    loop toggling feature'' of the interlace polynomial.}
  \label{fig:algo_example}
\end{figure}

The time bound of our algorithm stems from the following observation:
The number of parts managed at a node $i$ of the tree decomposition is
essentially bounded by the number of scenarios of its bag $X_i$. This
number is independent of the size of $G$ and single exponential in the
bag size (and thus single exponential in the treewidth of~$G$):

\begin{lem}
\label{lem:scenario_bound}
Let $U$ be an extension, i.e.\ a finite set of vertices, $|U|=k$.
There are less than $2^{(3k+1)k/2}$ scenarios of $U$.
\end{lem}
\begin{proof}
  The number of symmetric $\zo$-matrices of dimension $k\times k$ is
  $2^{(k+1)k/2}$, as a symmetric matrix is determined by its left
  lower half. Thus, there are $2^{(k+1)k/2}$ possibilities for
  $\sUU{U}$.

  For $\sUV{U}{V'}$, there less than $2^{k^2}$ possibilities: As there are
  $2^k-1$ non-zero elements of $\zo^k$, the number of linearly
  independent subsets of $\zo^U$ with $d$ elements is bounded by
  $\binom{2^k-1}{d}$. Thus, the number of \emph{all} linearly
  independent subsets of $\zo^U$ is at most
  \[ \sum_{0\leq d\leq k}\binom{2^k-1}{d} \leq (k+1)\binom{2^k-1}{k} <
  2^{k^2}. \quad\quad \quad \tag*{\qedhere}
\]
\end{proof}

\section{Symmetric Gaussian Elimination}
\label{sec:symgauss}

We want to convert adjacency matrices into matrices of a form as in
Figure~\ref{fig:scenario_intuition01} without touching the rank. In
order to achieve this, we introduce a special way of performing
Gaussian elimination that differs from standard Gaussian elimination
in the following way. First, it is symmetric, as in general every
column operation is followed by a corresponding row operation. In this
way, we maintain the correspondence between rows/columns of the matrix
we are manipulating and vertices of a graph. Second, we adhere to a
particular order when deciding which entry to use for the next pivot
operation. This order is (partially) fixed by the tree decomposition.
It is crucial for our proofs of the statements in
Sect.~\ref{sec:converting} that the elimination process proceeds
according to this order. Third, we perform symmetric Gaussian
elimination using only vertices in a \emph{subset} $V'$ of the
vertices: When seeking a pivot entry in a particular row/column, we do
not consider all entries of the row/column but only the ones that
correspond to edges between vertices in $V'$.

\subsection{Elimination with a Single Vertex}
\label{ssec:elimination_step}

Assume we are given a graph $G=(V,E)$, its adjacency matrix $M$ and a
vertex $v$. We would like to compute the rank of $M$ as the ``effect
of $v$ on the rank'' plus the rank of a submatrix in which we have
deleted $v$. This might not immediately be possible using $M$ itself,
but we can achieve it by an appropriate modification of $M$. Arratia
et al.\ observe that edge pivot and local complementation are such
appropriate modifications \cite[Lemma 2, Lemma
  5]{arratia_two_var_interl}. For our purposes, we want to control the
order of the operations on the adjacency matrix. Thus, we do not use
edge pivot and local complementation directly, but define a
\emph{symmetric Gaussian elimination step} on $M$ using $v$ in the
following way:
\begin{itemize}
\item If $v$ is an isolated vertex without a self loop, we have
  situation (1) of Figure~\ref{fig:steps}. Vertex $v$ has no influence
  on the rank of the adjacency matrix and we can delete the column and
  row corresponding to $v$ without changing the rank of the adjacency
  matrix. The result of the elimination step is just $M$.
\item If $v$ has a self loop, there is a $1$ in the $(v,v)$-entry of
  $M$. The elimination step consists of the following operations.
  Except for entry $(v,v)$, we remove all $1$s in the $v$-column and
  $v$-row using the following pair of operations for each neighbor $u$
  of $v$: First, add the $v$-column to the $u$-column.  Then, in the
  modified matrix, add the $v$-row to the $u$-row. We denote the
  result of the whole process by $M\se v$, which is depicted as (2) in
  Figure~\ref{fig:steps}.  Note that $M\se v$ is symmetric again.  The
  rank of $M$ equals $1$ plus the rank of $M\se v$ with $v$-column and
  $v$-row deleted.

  Note that -- up to order of the operations -- this is local
  complementation on $v$: Writing $\lc G v$ for the local complement
  of a graph $G$ on vertex $v$ \cite[Definition
    4]{arratia_two_var_interl}, the adjacency matrix of $\lc G v$ is
  $M \se v$ \cite[Proof of Lemma 5]{arratia_two_var_interl}.
\item If $v$ is neither isolated nor has a self loop, there is a
  neighbor $u$ of $v$. Assume that $u$ does not have a self loop. The
  $(u,v)$- and $(v,u)$-entries of $M$ equal $1$. The elimination step
  consists of the following operations. In the first stage, except for
  $(u,v)$ and $(v,u)$, we remove all $1$s in the $v$-column and
  $v$-row. This is accomplished by the following pair of operations
  for each neighbor $u'$ of $v$, $u'\neq u$: First, add the $u$-column
  to the $u'$-column.  Then, in the modified matrix, add the $u$-row
  to the $u'$-row.  Again, performing such a pair of column/row
  operations keeps a symmetric matrix symmetric.  At the end of the
  first stage the $v$-column and $v$-row consist entirely of $0$s,
  except for the entry at the $u$-position, which is~$1$. The second
  stage proceeds as follows: we add the $v$-column to every column
  which has a $1$ in the $u$-row, and we also add the $v$-row to every
  row which has a $1$ in the $u$-column. At the end of this stage also
  the $u$-column and $u$-row consist only of $0$s except at the
  $v$-position.  The result of the second stage is a symmetric matrix
  again, which we denote by $M\se vu$. It is depicted as (3) in
  Figure~\ref{fig:steps}. We do not swap columns/rows, as we must keep
  the vertices in a particular order, which is determined by the tree
  decomposition, cf.\ Section~\ref{ssec:order}. The rank of $M$ equals
  $2$ plus the rank of $M\se vu$ with $u$- and $v$-column and $u$- and
  $v$-row deleted. Note that this matrix is also the adjacency matrix
  of $\ep G {vu} [V\setminus \{v,u\}]$, where $\ep G {vu}$ denotes
  edge pivot of $G$ on vertices $v$ and $u$ \cite[Definition 1, Lemma
    2]{arratia_two_var_interl}.

  If $u$ has a self loop we proceed analogously to obtain a matrix
  with a structure as (4) in Figure~\ref{fig:steps}.  Then we can
  eliminate the self loop at $u$ by, say, adding column $v$ to column
  $u$. (As at this point column $v$ is zero everywhere except at $u$,
  only entry $(u,u)$ of the matrix is changed by this operation and
  the symmetry is not destroyed.) Thus, we obtain a matrix exactly as
  (3) in Fig.~\ref{fig:steps}.
\end{itemize}

\begin{figure}
  \centering
  \includegraphics{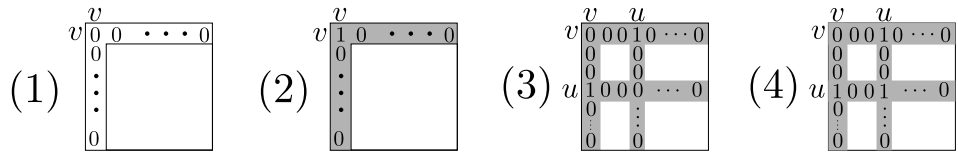}
  \caption{Effect of a symmetric Gaussian elimination step. Adjacency
    matrix with isolated unlooped vertex $v$ (1), adjacency matrix
    after eliminating with a self loop at $v$ (2), adjacency matrix
    after eliminating with edge $vu$ (3).}
  \label{fig:steps}
\end{figure}

We can describe the effect of a symmetric elimination step on the
entries of the matrix (aside from the entries being set to $0$) in the
following way.

\begin{lem}
  Let $M=(m_{ij})$ be an adjacency matrix, let $a$ be a vertex with a
  self loop, and $m_{yx}$ some entry of $M$ which is not in column or
  row $a$, i.e.\ $a\not \in \{x,y\}$.  Then, after symmetric Gaussian
  elimination using $a$, the $(y,x)$-entry of $M$ will be
  \[(M\se a)_{yx} = m_{yx} + m_{ax} m_{ya}. \]
\end{lem}

\begin{lem}
\label{lem:edgeel_effect}
Let $M=(m_{ij})$ be an adjacency matrix, let $a$ be a vertex without a
self loop, $ab$ an edge and $m_{yx}$ some entry of $M$ which is not in
column or row $a$ or $b$, i.e.\ $\{x,y\}\cut \{a,b\}=\emptyset$.
Then, after symmetric Gaussian elimination using $ab$, the
$(y,x)$-entry of $M$ will be
  \[(M\se ab)_{yx} = m_{yx} + m_{ax} m_{yb} + m_{ya} m_{bx} + m_{ax}
  m_{ya} m_{bb}. \]
\end{lem}

We prove the statement about edge elimination, the case of self loop
elimination is completely analogously.
\begin{proof}[of Lemma~\ref{lem:edgeel_effect}]
  Let us assume that $x\leq y$ (the case $x>y$ is analogous). The
  situation is depicted in Figure~\ref{fig:edgeel_effect}.  Depending
  on the $(a,x)$-entry being $1$ or not, column $b$ is added to column
  $x$, which adds the $(y,b)$-entry to the $(y,x)$-entry. This gives
  the term $m_{ax} m_{yb}$. After that, depending on the
  $(y,a)$-entry, row $b$ is added to row $y$. This adds the actual
  value of the $(b,x)$-entry to the $(y,x)$-entry. By the previous
  column addition, the actual $(b,x)$-entry is $m_{bx} + m_{ax}
  m_{bb}$. Thus, the row addition contributes a term $m_{ya}(m_{bx} +
  m_{ax}m_{bb})$. The second stage has no effect on the $(y,x)$ entry:
  Column $a$ may be added to some other columns. But at this point of
  time, column $a$ is entirely zero, except at the $b$ entry. Thus,
  addition of the $a$ column has no effect on the $(y,x)$ entry. The
  same is true for addition of the $a$ row.
\end{proof}

\begin{figure}
  \centering
  \includegraphics{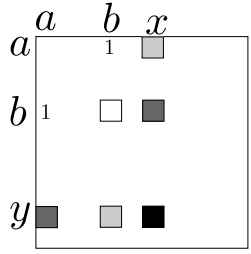}
  \caption{During a symmetric Gaussian elimination step using edge
    $ab$, entry $(y,x)$ is affected only by the entries at $(a,x)$,
    $(y,b)$, $(y,a)$, $(b,x)$ and $(b,b)$.}
  \label{fig:edgeel_effect}
\end{figure}

\subsection{Vertex Order, Elimination with Vertex Sets, and the
  Scenario of an Extended Graph}
\label{ssec:order}

We want to define symmetric Gaussian elimination using a whole set
$V'\subs V$ of vertices. This means that we perform elimination steps
using each vertex from $V'$. The result of this process depends on the
order in which we use the vertices for elimination steps. Therefore we
introduce an order on the vertices of the graph, which will be
computed before the computation of the interlace polynomial starts. We
will use this order throughout the rest of the paper.  Whenever there
could be any ambiguity, we proceed according to this order.

The vertex order we are using must be compliant with the tree
decomposition we are using: Whenever a vertex is forgotten, it must be
greater than all the vertices which have been forgotten before. Or,
equivalently, the vertices in the extension $X_i$ must be greater than
the vertices in $V_i$ for each node $i$ of the tree decomposition.
Such an order can be obtained by Algorithm~\ref{alg:order}.

\begin{algorithm}
\caption{Supplying a vertex order.}
\label{alg:order}
\begin{algorithmic}[1]
  \Procedure{SupplyVertexOrder}{}
  \State $c\gets 1$
  \ForAll {nodes $i$, in the order of bottom-up traversal, i.e.\ each father node is visited after all its children}
      \If {$i$ is a forget node}
        \State $a\gets$ vertex being forgotten at node $i$
        \State give vertex $a$ number $c$ in the vertex order
\label{alg:order:fix}
        \State $c\gets c+1$
      \EndIf
  \EndFor
  \EndProcedure
\end{algorithmic}
\end{algorithm}

Now we are ready to define elimination using a set of vertices.

\begin{defi}
\label{def:sge}
Let $V'\subs V$ be a set of vertices of a graph $G=(V,E)$ with
adjacency matrix $M$. \defexpr{Symmetric Gaussian elimination on $G$
  using $V'$} is defined as the following process: If $V'=\emptyset$,
we are done and $M$ is the output of the symmetric Gaussian
elimination process using $V'$. Otherwise, we let $v$ be the minimum
vertex in $V'$. If $v$ has a self loop we let $M'=M\se v$. Otherwise,
we check whether $v$ has a neighbor $u$ in $V'$.  If yes, we let
$M'=M\se vu$, where $u$ is the minimum neighbor of $v$.  If no, we let
$M'=M$.  This concludes the processing of $v$. To complete the
elimination using $V'$, we continue recursively with $V'\setminus
\{v\}$ in the role of $V'$ and $M'$ in the role of $M$.
\end{defi}

We also order vertex vectors (i.e.\ elements from $\zo^U$, $U$ some
vertex set) and sets of vertex vectors according to the vertex order
(lexicographically). This induced order is used for choosing a
``minimal'' basis in the following definition.

\begin{defi}[Scenario of an extended graph]
\label{def:extension_scenario}
Let $\eG{V'}{U}$ be an extended graph obtained by extending $G[V']$ by
$U$ according to graph $G=(V,E)$. Let the vertex order be such that
$v'<u$ for all $v'\in V'$ and $u\in U$. Then the scenario
$\Sc(\eG{V'}{U})$ of $\eG{V'}{U}$ is defined as follows: Let $M$ be
the adjacency matrix of $G[V'\uni U]$. Perform symmetric Gaussian
elimination on $M$ using $V'$ to obtain $M'$. Let $M'_{UV'}$ be the
$U\times V'$ submatrix of $M'$. Consider the column space $W$ of
$M'_{UV'}$. We can choose a basis of $W$ from the column vectors of
$M'_{UV'}$. Let $\sUV{U}{V'}$ be the minimal such basis.  Let
$\sUU{U}$ be the contents of the $U\times U$ submatrix of $M'$. We
define $\Sc(\eG{V'}{U})$ to be $(\sUV{U}{V'}, \sUU{U})$.
\end{defi}

The minimal basis $\sUV{U}{V'}$ in the preceding definition can by
obtained by the following steps: Start with an empty set of columns
and then as often as possible take the minimum column of $M'_{UV'}$
which is not in the span of the so far collected columns.

\section{Scenarios and Nice Tree Decompositions}
\label{sec:scenarios}
\label{sec:converting}

\begin{figure}
  \centering
  \includegraphics{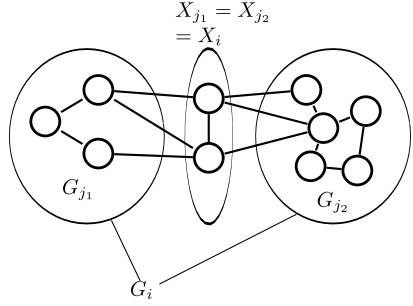}  
  \caption{Graphs corresponding to a join node $i$ and its child nodes
    $j_1$, $j_2$.}
  \label{fig:join_node}
\end{figure}

Consider a join node $i$ with children $j_1$ and $j_2$ in a nice tree
decomposition of a graph $G$ the interlace polynomial of which we want
to evaluate. By the properties of tree decompositions, this implies a
situation as depicted in Figure~\ref{fig:join_node}:
$G_{j_1}=G[V_{j_1}]$ and $G_{j_2}=G[V_{j_2}]$ are disjoint graphs with
a common extension $X_{j_1}=X_{j_2}=X_i$. $G_i=G[V_i]=G[V_{j_1}\uni
V_{j_2}]$ is the disjoint union of $G_{j_1}$ and $G_{j_2}$. Assume
that we have computed all parts (see Section~\ref{sec:ideas} and
\eqref{eq:bubble_term}) of the interlace polynomial of $G_{j_1}$ and
all parts of the interlace polynomial of $G_{j_2}$. From this we want
to compute the parts of the interlace polynomial of $G_i$.  Consider
one such part, say the one corresponding to some scenario $s$ of
$X_i$. Somehow we have to find out for which subgraphs\footnote{In
  fact induced subgraphs with self loops toggled at some vertices ---
  but we will ignore this detail for the rest of the section as it is
  not important to understand the idea.}  $G[V']$ of $G_i$ the
scenario of the extended graph $\eG{V'}{X_i}$ is $s$.  Fortunately,
these are exactly the subgraphs $G[V_1\uni V_2]$, $V_1\subs V_{j_1}$,
$V_2\subs V_{j_2}$, with the property that the ``join'' of the
scenario of $\eG{V_1}{X_{j_1}}$ and the scenario of
$\eG{V_2}{X_{j_2}}$ is $s$. This is guaranteed by the following lemma.

\begin{lem}[Join]
\label{lem:rank_join}
Let $G=(V,E)$ be a graph, $U\subs V$, and $s_1, s_2$ two scenarios of
$U$. Then there is a unique scenario $s_3$ of $U$ such that the
following holds: If $G[V_1]$ and $G[V_2]$ are disjoint subgraphs of
$G$ that may be extended by $U$ according to $G$, $\Sc(\eG{V_1}{U}) =
s_1$, and $\Sc(\eG{V_2}{U}) = s_2$, then $\Sc(\eG{V_1 \uni V_2}{U}) =
s_3$. Moreover, $s_3$ can be computed from $s_1, s_2$ and $G[U]$
within $\poly(|U|)$ steps.
\end{lem}

\begin{figure}
  \centering
  \includegraphics{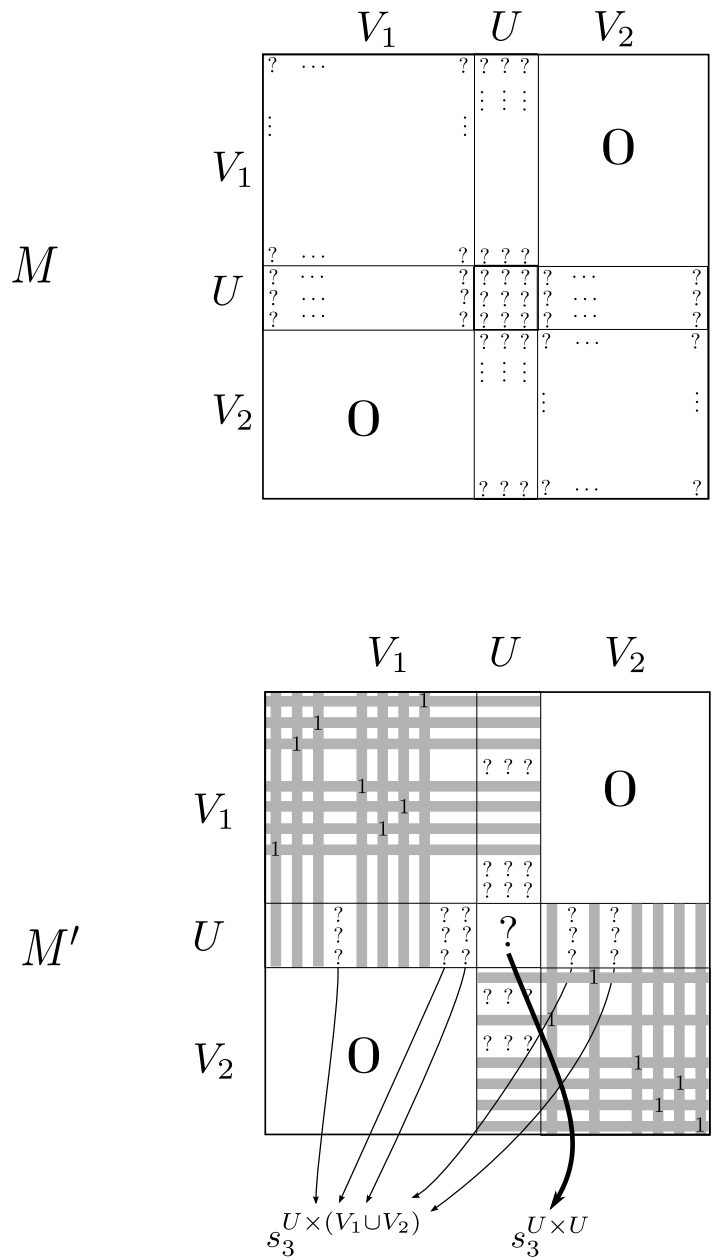}
  \caption{Effect of symmetric Gaussian elimination to gain the
    scenario of $\eG{V_1\duni V_2}{U}$. Entries with question marks
    are either $0$ or $1$. Empty entries are $0$.}
  \label{fig:Joinrank}
\end{figure}

\begin{proof} 
  We will apply Definition~\ref{def:extension_scenario} to determine
  $s_3$. We will see that $s_3$ is uniquely defined by $s_1$, $s_2$
  and $G[U]$, and can be computed from these within the claimed time
  bound. This will prove the lemma.

  Let $G_1=G[V_1]$ and $G_2=G[V_2]$. Let $M$ be the adjacency matrix
  of $G[V_1\uni V_2\uni U]$. As $G_1$ and $G_2$ are disjoint, $M$ has
  a form as depicted on the left hand side in
  Figure~\ref{fig:Joinrank}, the $V_1 \times V_2$ submatrix as well as
  the $V_2\times V_1$ submatrix of $M$ consists only of $0$s.

  By Definition~\ref{def:extension_scenario}, symmetric Gaussian
  elimination using $V_1\uni V_2$ has to be performed on $M$ to obtain
  $M'$, which is of the form depicted on the right hand side in
  Figure~\ref{fig:Joinrank} and from which $s_3$ can be read off.  Let
  us analyze a single elimination step occurring during the
  elimination process in detail, say eliminating with a self loop at a
  vertex $v \in V_1$.  One action in this step is that the $1$ in the
  $(v,v)$ entry will be used to eliminate another $1$ in the $v$-row
  by adding the $v$-column to the respective column $u$. Let us argue
  that this affects neither the $V_1\times V_2$ submatrix of $M$ nor
  the $V_2\times V_1$ submatrix of $M$. As $v\in V_1$, in the $v$-row
  the $V_2$-entries are already $0$. Thus we know that $u\not \in
  V_2$, i.e.  the $v$-column will be added to a column from $V_1\uni
  U$.  Thus, the $V_1\times V_2$ submatrix is not changed. Again as
  $v\in V_1$, the $V_2$-entries in the $v$-column are $0$ and addition
  of the $v$-column to any other column $u$ does not change the
  $V_2$-entries of column $u$. Thus, the $V_2\times V_1$ submatrix of
  $M$ is not changed.

  Analogous observations can be made for the role of columns and rows
  reversed (i.e.\ when adding the $v$-row to other rows to eliminate
  $1$s in the $v$-column), as well as for elimination steps using an
  edge between different vertices (instead of self loops). We conclude
  that symmetric Gaussian elimination steps with $V_1$-vertices affect
  only the $(V_1\uni U)\times (V_1 \uni U)$ submatrix of $M$, but not
  the $V_1\times V_2$ or $V_2\times V_1$ submatrix.  Analogously,
  elimination steps with $V_2$-vertices affect only the $(U\uni
  V_2)\times (U\uni V_2)$ submatrix of $M$.  Thus, except for the
  $U\times U$ submatrix, when performing symmetric Gaussian
  elimination on $M$ using $V_1\uni V_2$, the same things happen as
  when performing symmetric Gaussian elimination first on $G[V_1\uni
  U]$ using $V_1$ and then on $G[V_2\uni U]$ using $V_2$.  The only
  difference may be that depending on the vertex order elimination
  steps with $V_1$-vertices are interlaced with steps using $V_2$
  vertices. But we argued that $V_1$-elimination steps do not
  influence parts of $M$ which are relevant for $V_2$-elimination
  steps and vice versa, so this is not an issue.

  As elimination on $M$ using $V_1\uni V_2$ (yielding $M'$) on the one
  hand does the same as elimination on $G[V_1\uni U]$ using $V_1$
  (yielding, say, $M^{(1)}$) and elimination on $G[V_2\uni U]$ using
  $V_2$ (yielding, say, $M^{(2)}$) on the other hand, the $U\times
  (V_1\uni V_2)$ submatrix of $M'$ is just the union of
  $M^{(1)}_{UV_1}$, the $U\times V_1$ submatrix of $M_1$, and
  $M^{(2)}_{UV_2}$, the $U\times V_2$ submatrix of $M_2$. Recall that
  $\SUV{s_1}{U}{V_1}$ and $\SUV{s_2}{U}{V_2}$ are minimum bases of the
  column space of $M^{(1)}_{UV_1}$, $M^{(2)}_{UV_2}$, resp, taken from
  the columns of these matrices. To compute $\SUV{s_3}{U}{(V_1\uni
    V_2)}$, the minimum basis of the column space of the $U\times
  (V_1\uni V_2)$ submatrix of $M'$ taken from the columns of this
  matrix, we proceed in the following way: Start with the empty set
  and as long as possible add the minimum vector of $\SUV{s_1}{U}{V_1}
  \uni \SUV{s_2}{U}{V_2}$ which is not in the span of the so far
  collected vectors. This can be done in time polynomial in $|U|$
  using standard Gaussian elimination.

  The $U\times U$ submatrix is the only part of $M$ which is affected
  by both, eliminations with $V_1$-vertices and eliminations with
  $V_2$-vertices. However, the use of the $U\times U$ submatrix is
  ``write-only'' during the elimination process: Consider symmetric
  Gaussian elimination in general, say on some extended graph
  $G[V'\uni U]$ using $V'$. Recall that by Definition~\ref{def:sge}
  all the elimination steps will involve only vertices from $V'$ in
  the sense that the step is either $M\se v$ or $M\se vu$ with
  $u,v\in V'$. Thus, the contents of the $U\times U$ submatrix has no
  influence on what elimination steps will be performed. All that
  happens with this submatrix is that column/row vectors are added to
  it.

  Thus, the effect on the $U\times U$ submatrix of all the elimination
  steps during symmetric Gaussian elimination of $G[V_1\uni U]$ using
  $V_1$ can be described as adding a matrix, say $A_1$ to the
  adjacency matrix of $G[U]$. We can compute $A_1$ as $A_1 =
  \SUU{s_1}{U} - M(G[U])$, where $M(G[U])$ denotes the adjacency
  matrix of $G[U]$.  Analogously, we can compute $A_2$ which describes
  the effect of symmetric Gaussian elimination of $G[V_2\uni U]$ using
  $V_2$ on the $U\times U$ submatrix. Because of the ``write-only''
  property, the effect of symmetric Gaussian elimination of $M$ using
  $V_1\uni V_2$ on the $U\times U$ submatrix of $M$ can be described
  by $A_1+A_2$.  Thus we have $\SUU{s_3}{U}=M(G[U])+A_1+A_2$, which is
  the second component of $s_3$.
 \end{proof}

\begin{defi}
\label{def:scenario_join}
  In the situation of Lemma~\ref{lem:rank_join} we write $\st{s_1}{s_2}{G[U]}$
  for $s_3$.
\end{defi}

To handle join nodes of the tree decomposition we proved
Lemma~\ref{lem:rank_join}: from the scenario of two extended graphs
$\eG{V_1}{U}$ and $\eG{V_2}{U}$ with a common extension $U$ we can
compute the scenario of the joined extended graph $\eG{V_1\uni
  V_2}{U}$ (cf.\ Figure~\ref{fig:join}). To handle also introduce and
forget nodes we prove two more lemmas (cf.\
Figure~\ref{fig:introduce}, Figure~\ref{fig:forget}).

\begin{figure}
  \centering
  \includegraphics{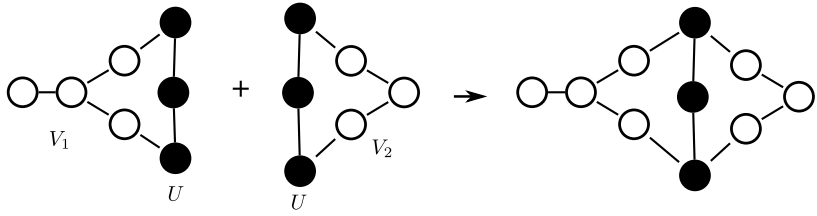}
  \caption{Joining the extended graphs $\eG{V_1}{U}$ and $\eG{V_2}{U}$.}
  \label{fig:join}
\end{figure}
\begin{figure}
  \centering
  \includegraphics{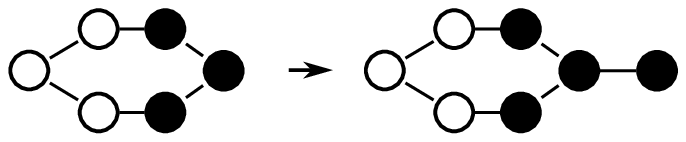}
  \caption{Adding a vertex to an extension.}
  \label{fig:introduce}
\end{figure}
\begin{figure}
  \centering
  \includegraphics{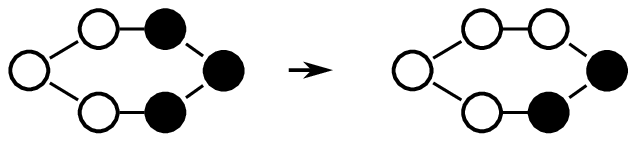}
  \caption{Transforming an extending vertex into a normal vertex.}
  \label{fig:forget}
\end{figure}

\begin{lem}[Introduce vertex]
\label{lem:rank_introduce}
Let $G=(V,E)$ be a graph, $U\subs V$, $s$ a scenario of $U$, $u\in
V\setminus U$.  Then there is a unique scenario $\s$ of $\UU=U\uni
\{u\}$ such that the following holds: If $G[V']$ may be extended by
$\UU$ according to $G$, $u$ is not connected to $V'$ in $G$, and
$\Sc(\eG{V'}{U}) = s$, then $\Sc(\eG{V'}{\UU}) = \s$. Moreover, $\s$
can be computed from $s$ and $G[\UU]$ in $\poly(|U|)$ steps.
\end{lem}

\begin{proof} 
  As $u$ is not connected to $V'$, $\SUV{\s}{\UU}{V'}$ is
  $\sUV{U}{V'}$ with a zero component for $u$ added to all the basis
  vectors. Also, $\SUU{\s}{\UU}$ is just $\sUU{U}$ with a row and
  column added representing the neighbors of $u$ in $\UU$.
 \end{proof}

\begin{defi}
\label{def:scenario_introduce}
In the situation of Lemma~\ref{lem:rank_introduce} we write
$\sp{s}{u}{G[\UU]}$ for $\s$.
\end{defi}

Except for isolated vertices without self loops, every vertex has an
effect on the rank of the adjacency matrix \cite[Lemma 2, Lemma
  5]{arratia_two_var_interl}
(cf. Section~\ref{ssec:elimination_step}). The following lemma states
that this effect can be extracted from the scenario.
\begin{lem}[Forget vertex]
\label{lem:rank_forget}
Let $G=(V,E)$ be a graph, $u\in U\subs V$, $\UU=U\setminus \{u\}$,
$\VV=V'\uni \{u\}$, and $s$ a scenario of $U$.  Then there is a unique
scenario $\s$ of $\UU$ and $r \in \zot$, $n\in \{1,0,-1\}$ such that
the following holds: If $G[V']$ is a subgraph of $G$ that may be
extended by $U$ according to $G$, $u>v'$ for all $v'\in V'$, and
$\Sc(\eG{V'}{U})=s$, then $\Sc(\eG{\VV}{\UU}) = \s$ and the rank
(nullity) of the adjacency matrix of $G[\VV]$ equals the rank
(nullity, resp.) of the adjacency matrix of $G[V']$ plus $r$ ($n$,
resp.). Moreover, $\s$ and $r$ can be computed from $s$ and $G[U]$ in
$\poly(|U|)$ steps, and we have $n=1-r$.
\end{lem}

\begin{figure}
  \centering
  \includegraphics{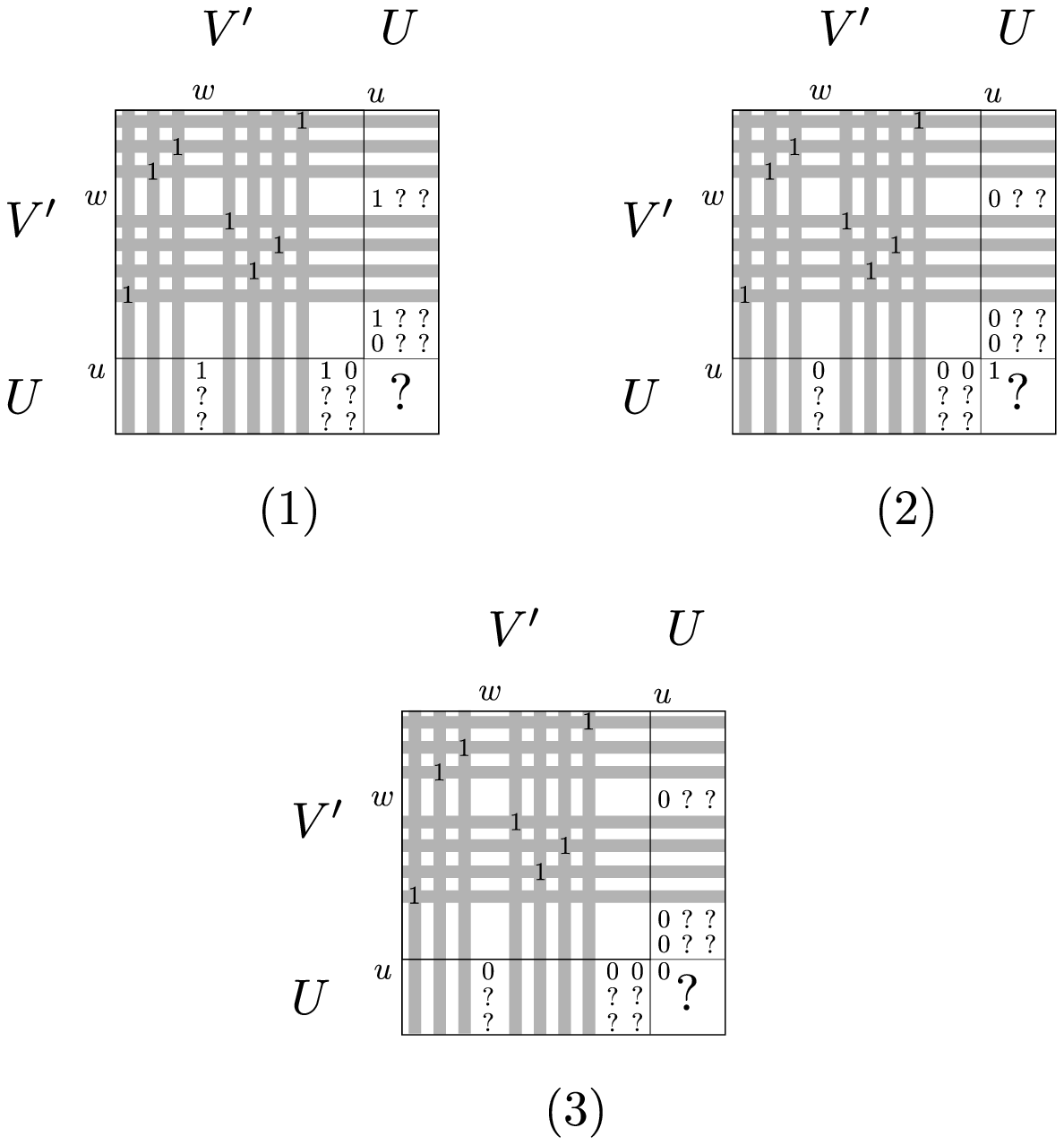}
  \caption{Cases when ``forgetting'' an extension vertex $u$. Entries
    with question marks are either $0$ or $1$. Empty entries are $0$.}
  \label{fig:Forgetrank}
\end{figure}

\begin{proof} 
  Consider the situation after symmetric Gaussian elimination on
  $G[V'\uni U] = G[\VV\uni \UU]$ using $V'$
  (Figure~\ref{fig:Forgetrank}). We distinguish three cases: (1) there
  is a basis vector of the $(U\times V')$ column space with a $1$ in
  the $u$-component, (2) there is no such basis vector, but the
  $(u,u)$-entry of the $U\times U$ submatrix equals 1, (3) neither
  case (1) nor (2).

  Let us first consider cases (2) and (3). As all $u$-components of
  the vectors in $\sUV{U}{V'}$ are zero, we know that symmetric Gaussian
  elimination on $G[\VV\uni \UU]$ using $\VV$ will consist of the
  following two stages: first, exactly the same operations will be
  performed as in symmetric Gaussian elimination on $G[V'\uni U]$
  using $V'$ (which will end up in the situations depicted in
  Figure~\ref{fig:Forgetrank} (2), (3)), and then elimination using
  vertex $u$ will be performed if possible.

  Thus, in case (3), $\s$ can be obtained from $s$ in the following
  way: remove the $u$ component of each vector of $\sUV{U}{V'}$ to
  gain $\SUV{\bar{s}}{\UU}{\VV}$. Let $a$ be the first column of
  $\sUU{U}$. Remove the first component of $a$. With standard Gaussian
  elimination, check in time $\poly(|U|)$ if $a$ is in the span of
  $\SUV{\bar{s}}{\UU}{\VV}$. If it is, let
  $\SUV{\s}{\UU}{\VV}=\SUV{\bar{s}}{\UU}{\VV}$, otherwise let
  $\SUV{\s}{\UU}{\VV}=\SUV{\bar{s}}{\UU}{\VV}\uni \{a\}$. Let
  $\SUU{\s}{\UU}$ be $\sUU{U}$ with first column and first row
  deleted. We have $r=0$ and $n=1$.

  In case (2), we first perform an elimination step with the 1 at the
  $(u,u)$-entry: let $\SUU{\bar{s}}{U} = \sUU{U}\se u$. Then we
  continue as in case (3) but with $\SUU{\bar{s}}{U}$ in the role of
  $\sUU{U}$. We have $r=1$ and $n=0$.

  The rest of this proof deals with case (1). Let $w\in V'$ be the
  vertex corresponding to the minimum vector of $\sUV{U}{V'}$ with a $1$
  in the $u$-component (cf.\ Figure~\ref{fig:Forgetrank} (1)). Compare
  symmetric Gaussian elimination on $G[V'\uni U]$ using $V'$ (which is
  performed to obtain $s$) to symmetric Gaussian elimination on
  $G[\VV \uni \UU]$ using $\VV$ (which is performed to obtain
  $\s$). Before these two processes reach $w$, they are equal, but
  from $w$ on they will differ: Using $V'$, the edge $uw$ will not be
  used for elimination and the process will continue with the next
  vertex in $V'$ immediately. Using $\VV$, the edge $uw$ will be used
  for elimination (which will not affect the $V'\times V'$ submatrix,
  but possibly change the contents of the $U\times (V'\uni U)$ and the
  $(V'\uni U) \times U$ submatrices).  Only after that, the process
  will continue with the next vertex in $\VV$. However, we will prove
  in Lemma~\ref{lem:uw_displace} that we can defer the elimination
  using edge $uw$ until all vertices of $V'$ have been proceeded and
  still obtain $\s$. Thus, $\s$ can be computed in the following way:
  perform the same steps as with symmetric Gaussian elimination on
  $G[V'\uni U]$ using $V'$. Then, simulate the effect of a symmetric
  Gaussian elimination step using edge $uw$ in a similar way as in
  cases (2) and (3).

  This simulation can be done as follows: Let $\vec{w}$ be the minimum
  vector of $\sUV{U}{V'}$ with the $u$-component equal to $1$. Let
  $\SUV{\bar{s}}{U}{V'} = \sUV{U}{V'}\setminus \{\vec{w}\}$ and
  $\SUU{\bar{s}}{U}=\sUU{U}$. For each row $i$, $i\neq u$, with the
  $\vec{w}_i=1$ simulate addition of column/row $u$ to column/row $i$ doing
  the following:
  \begin{enumerate}
  \item For each vector $\vec{c}$ of $\SUV{\bar{s}}{U}{V'}$, add
    component $u$ of $\vec{c}$ to component $i$ of $\vec{c}$.
  \item Change $\SUU{\bar{s}}{U}$ by first adding the $u$ column to
    the $i$ column and then, in the modified matrix, the $u$ row to
    the $i$ row.
  \end{enumerate}
  We have $\SUV{\s}{\UU}{\VV}=\SUV{\bar{s}}{U}{V'}$, and
  $\SUU{\s}{\UU}$ is $\SUU{\bar{s}}{U}$ with first column and first
  row removed. Note that after an elimination step using edge $wu$,
  the $u$ column/row will consist entirely of zeros (except at $(u,w)$
  and $(w,u)$). Thus, the first column of $\SUU{\bar{s}}{U}$ will be
  zero after the elimination with $wu$ and we do not need to
  incorporate it into $\SUV{\s}{\UU}{\VV}$.

  Finally note that we have $r=2$ and $n=-1$ in case (1).
  
\end{proof}

\begin{defi}
\label{def:scenario_forget}
In the situation of Lemma~\ref{lem:rank_forget} we write
$\sm{s}{u}{G[U]}$ for $\s$, $\dr{s}{u}{G[U]}$ for $r$, and
$\dn{s}{u}{G[U]}$ for $n$.
\end{defi}

The operation defined in Definition~\ref{def:scenario_forget} deletes
a vertex $u$ from a scenario in the sense that $u$ is deleted from the
extension but added to the graph being extended. We also need a
notation for deleting a vertex completely from a scenario, i.\ e.\
ignoring some vertex of the extension.

\begin{defi}
\label{def:scenario_ignore}
Let $s=(\sUV{U}{V'},\sUU{U})$ be a scenario of an extension $U$ and
$u\in U$. Then $\si{s}{u}$ is the scenario obtained from $s$ in the
following way: Delete the $u$-components from the elements of
$\sUV{U}{V'}$ to obtain $s_1$. Choose the minimum (according to the
vertex order) basis $s_1'$ for the span of $s_1$ from the elements of
$s_1$ using standard Gaussian elimination.  Delete the $u$-column and
$u$-row from $\sUU{U}$ to obtain $s_2$. We define
$\si{s}{u}=(s_1',s_2)$.
\end{defi}

The following lemma is used in the proof of
Lemma~\ref{lem:rank_forget}.

\begin{lem}
\label{lem:uw_displace}
Let $G=(V,E)$ be a graph, $u\in U\subs V$ and $G'=G[V']$ a subgraph of
$G$ which may be extended by $U$ and $u>v'$ for all $v'\in V'$. Let
$w$ be the minimum vertex of $V'$ and assume that $u$ is the minimum
neighbor of $w$ (which implies that $w$ has no neighbor in $V'$). Let
$V''=V'\uni\{u\}$, $\VV = V'\setminus \{w\}$ and $M$ be the adjacency
matrix of $G[V'\uni U]$ (cf.\ Figure~\ref{fig:displace}).  Then the
following two sequences of operations on $M$ lead to the same result:
\begin{enumerate}
\item Symmetric Gaussian elimination on $M$ using $V''$, i.e.\ first
  the elimination step using edge $wu$ and then the elimination steps
  using $\VV$.
\item Symmetric Gaussian elimination on $M$ using $V'$ (i.e.\ the
  elimination steps using $\VV$, as $w$ has no neighbor in $V'$) and
  after that, on the result, the elimination step using edge $wu$.
\end{enumerate}
\end{lem}

\begin{figure}
  \centering
  \includegraphics{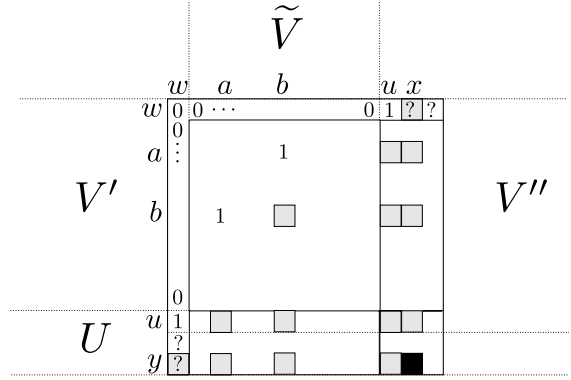}
  \caption{Symmetric Gaussian elimination using $\VV$ (including steps
    such as eliminating with edge $ab$) and eliminating with edge $wu$
    can be swapped without changing the result. Empty entries and
    entries with a question mark are either $0$ or $1$.}
  \label{fig:displace}
\end{figure}

\begin{proof}
  Elimination with edge $wu$ will add the $u$ column (row, resp.) to
  all columns (rows, resp.) which have a $1$ in the $w$-row (column,
  resp.), and will then eliminate any remaining $1$ in the $u$ column
  (row, resp.). As the $V'$-part of the $w$ row (column., resp.) is
  entirely zero, this has no influence on the $\VV\times \VV$
  submatrix of $M$. Thus, the only difference between 1.\ and 2.\ is
  whether the elimination step using edge $wu$ is performed before or
  after symmetric Gaussian elimination using $\VV$. Also, it is
  enough to consider the $U$-columns and $U$-rows of $M$. We will
  ignore the $V'\times V'$ submatrix of $M$ in the following.

  We will prove the following: every elimination step using an edge
  $ab$ (a self loop at $a$, resp.) in $\VV$ can be swapped with
  elimination using $wu$, i.e.\ the results of $\se{ab}\se{wu}$ and
  $\se{wu}\se{ab}$ ($\se{a}\se{wu}$ and $\se{wu}\se{a}$, resp.) are
  equal. Applying this observation repeatedly proves the lemma. We
  only prove the case of an edge $ab$ in $\VV$, the case of a self
  loop at $a$ in $\VV$ can be dealt with similarly.

  Let $ab$ an edge in $\VV$. First, let us consider the column and
  rows of $a$, $b$, $w$ and $u$. It is not hard to see, that, no
  matter whether we use first $ab$ for elimination and then $wu$ or
  vice versa, in the end these columns will consist entirely of zeros,
  except for $(u,w)$, $(w,u)$, $(a,b)$, $(b,a)$. Thus, it is
  sufficient to examine the effect of both elimination steps on
  entries $(y,x)$ with $\{x,y\}\cut \{a,b,u,w\} = \emptyset$, cf.\
  Figure~\ref{fig:displace}.

  Let $M^{ab}=M\se ab$ be $M$ after the elimination step using edge
  $ab$. Analogously we let $M^{wu}=M\se wu$, as well as
  $M^{ab,wu}=M\se ab \se wu$ and $M^{wu,ab}=M\se wu \se ab$. We use
  small $m$ to denote the entries of these matrices. For instance,
  $m_{yx}^{ab,wu}$ denotes the entry in row $y$ and column $x$ of
  $M^{ab,wu}$.

  \begin{description}
  \item[Case ``$ab$ first''.] By Lemma~\ref{lem:edgeel_effect} we have
    \[m_{yx}^{ab} = m_{yx} + m_{ax} \cdot m_{yb} + m_{ya} \cdot m_{bx}
    + m_{ya}\cdot m_{ax}\cdot m_{bb}.\] By
    Lemma~\ref{lem:edgeel_effect} again, the final value of entry
    $(y,x)$ is
\begin{equation*}
  \label{eq:abfirst_final}
  m_{yx}^{ab,wu} = m_{yx}^{ab} + m_{wx}^{ab}\cdot m_{yu}^{ab} +
  m_{yw}^{ab}\cdot m_{ux}^{ab} + m_{wx}^{ab} \cdot m_{yw}^{ab} \cdot m_{uu}^{ab},
\end{equation*}
where $m_{wx}^{ab} = m_{wx}$ and $m_{yw}^{ab}=m_{yw}$, as the
elimination using edge $ab$ does not affect column/row $w$
(cf.\ Figure~\ref{fig:displace}).  Furthermore, we have
\begin{align*}
m_{yu}^{ab} &= m_{yu} + m_{au} \cdot m_{yb} + m_{ya} \cdot m_{bu} + m_{au} \cdot m_{ya} \cdot m_{bb}, \\
m_{ux}^{ab} &= m_{ux} + m_{ax} \cdot m_{ub} + m_{ua} \cdot m_{bx} + m_{ax} \cdot m_{ua} \cdot m_{bb}, \\
m_{uu}^{ab} &= m_{uu} + m_{au} \cdot m_{ub} + m_{ua} \cdot m_{bu} + m_{au} \cdot m_{ua} \cdot m_{bb},
\end{align*}
once more by Lemma~\ref{lem:edgeel_effect}.

  \item[Case ``$wu$ first''.] Here we have
\begin{equation*}
  \label{eq:uwfirst_final}
  m_{yx}^{wu,ab} = m_{yx}^{wu} + m_{ax}^{wu}\cdot m_{yb}^{wu} +
  m_{ya}^{wu}\cdot m_{bx}^{wu} + m_{ax}^{wu}\cdot m_{ya}^{wu} \cdot m_{bb}^{wu},
\end{equation*}
where $m_{bb}^{wu}=m_{bb}$, as the entry $(b,b)$ is not affected by
edge elimination using edge $wu$. For the remaining values we have by
Lemma~\ref{lem:edgeel_effect}:
\begin{align*}
  m_{yx}^{wu} &= m_{yx} + m_{wx} \cdot m_{yu} + m_{yw} \cdot m_{ux} + m_{wx} \cdot m_{yw} \cdot m_{uu}, \\
  m_{ax}^{wu} &= m_{ax} + m_{wx} \cdot m_{au}, \\
  m_{yb}^{wu} &= m_{yb} + m_{yw} \cdot m_{ub}, \\
  m_{ya}^{wu} &= m_{ya} + m_{yw} \cdot m_{ua}, \\
  m_{bx}^{wu} &= m_{bx} + m_{wx} \cdot m_{bu}.
\end{align*}
\end{description}

An easy calculation yields that $m_{yx}^{wu,ab}=m_{yx}^{ab,wu}$, which
completes the proof.
 \end{proof}

\section{The Algorithm}
\label{sec:algo}
\label{sec:blubble_term}

Algorithm~\ref{alg:main} evaluates the interlace polynomial using a
tree decomposition. The input for the algorithm is $G=(V,E)$, the
graph of which we want to evaluate the interlace polynomial, and a
nice tree decomposition $(\{X_i\}_I, (I,F))$ of $G$ with $O(n)$ nodes,
$n=|V|$. In Section~\ref{ssec:treedecomp} we discussed how to obtain a
nice tree decomposition. Let $k-1$ be the width of the tree
decomposition, i.e.\ $k$ is the maximum bag size.

\begin{algorithm}
\caption{Evaluating the interlace polynomial using a tree decomposition.}
\label{alg:main}
\label{alg:parts}
\begin{algorithmic}[1]
  \Input Graph $G$, nice tree decomposition $(\{X_i\}_i, (I, F))$ of
  $G$, $k$ such that any bag $X_i$ of the tree decomposition contains
  at most $k$ vertices
\label{alg:main:emptyroot}
  \State \Call{SupplyVertexOrder}{} \Comment{Algorithm~\ref{alg:order}}
\label{alg:main:order}
\label{alg:main:m1}
\ForAll {nodes $i$ of the tree decomposition, in the order they appear in bottom-up traversal}
\ForAll {$D\subs X_i$}
          \If {$i$ is a leaf}
            \State $S(i,D, ((),())) \gets 1$
\label{alg:parts:leaf}
\label{alg:main:leaf}
          \ElsIf {$i$ is a join node}
            \State \Call{Join}{$i, D$}
          \ElsIf {$i$ is an introduce node}
            \State \Call{Introduce}{$i, D$}
          \ElsIf {$i$ is a forget node}
            \State \Call{Forget}{$i, D$}
          \EndIf
\label{alg:main:compparts}
\label{alg:main:m2}
\EndFor
\EndFor
\State return $S(\mathrm{root}, \emptyset, ((),()))$ \Comment
{$X_\mathrm{root} = \emptyset$}
\label{alg:main:return}
\end{algorithmic}
\end{algorithm}

\subsection{Interlace Polynomial Parts}

Algorithm~\ref{alg:main} essentially traverses the tree decomposition
bottom-up and computes parts $S(i,D,s)$ of the interlace polynomial
for each node $i$. For a node $i$, $D\subs X_i$, and a scenario $s$ of
$X_i$, one such part is defined in the following way:
\begin{equation}
  \label{eq:bubble_term}
    S(i, D, s) = \sum_{A,B} x_A y_B u^{\rk((G_i\slt B)[A\uni B])} v^{\n((G_i\slt B)[A\uni B])},
\end{equation}
where the summation extends over all $A, B \subs V_i$ with $A\cut B =
\emptyset$ and \[\Sc(\EG{G'}{A\uni B}{X_i}) = s,\quad G'=G\slt (B\uni
D).\] Recall that $V_i$ is the set of vertices which have been
forgotten below node $i$. Thus, $S(i,D,s)$ is the part of the
interlace polynomial of $G[V_i]$ corresponding to $D$ and $s$.


For every leaf $i$ of the tree decomposition we have $V_i=\emptyset$
and also $X_i=\emptyset$. Thus, in Line~\ref{alg:main:leaf} of
Algorithm~\ref{alg:main} we have $D=\emptyset$.  Trivially,
$\Sc(\eG{\emptyset}{\emptyset})$ is the empty scenario. Thus, we have
$S(i,\emptyset,((),()))=1$ if $i$ is a leaf.

At the root node $r$ the bag $X_r$ is empty and all vertices have been
forgotten, i.e.\ $V_r=V$. There is only one part left, $S(r,
\emptyset, ((), ())$, and this is just the interlace polynomial of
$G$.

\subsection{Join Nodes}

Join nodes are handled by Algorithm~\ref{alg:join}. The correctness
follows from
\begin{lem}
\label{lem:sum_join}
  Let $i$ be a join node with children $j_1$ and $j_2$,
  $D\subs X_i$ and $s$ a scenario of $X_i$. Then
  \begin{equation}
    \label{eq:sum_join}
    S(i,D, s) = \sum_{s_1, s_2} S(j_1, D, s_1) S(j_2, D, s_2),
  \end{equation}
  where the summation extends over all scenarios $s_1, s_2$ of $X_i$
  such that \[\st{s_1}{s_2}{G\slt D [X_i]} = s.\]
\end{lem}

\begin{proof}
  Recall \eqref{eq:bubble_term} for node $i$. Every admissible $A,B$
  give rise to $A_1=A\cut V_{j_1}$, $A_2=A\cut V_{j_2}$, $B_1 = B\cut
  V_{j_1}$, $B_2 = B\cut V_{j_2}$. $G'[A\uni B]$ is the disjoint union
  of $G'[A_1\uni B_1]$ and $G'[A_2 \uni B_2]$. (These graphs are
  subgraphs of the ones depicted in Figure~\ref{fig:join_node}.)

  We can apply Lemma~\ref{lem:rank_join} with $G'$ in the role of $G$,
  $A_1\uni B_1$ in the role of $V_1$ and $A_2 \uni B_2$ in the role of
  $V_2$. This implies that $A\uni B$ takes the role of $V_1 \uni V_2$.
  Using this it is not hard to argue that every admissible $(A,B)$ in
  \eqref{eq:bubble_term} corresponds to one pair $((A_1,
  B_1),(A_2,B_2))$ of the expanded version of \eqref{eq:sum_join}.
 \end{proof}

\begin{algorithm}
\caption{Computing the parts at a join node.}
\label{alg:join}
\begin{algorithmic}[1]
  \Procedure{Join}{$i$, $D$}
  \ForAll {scenarios $s$ for $|X_i|$ vertices}
\label{alg:join:init}
  \State {} \Comment {i.e., enumerate all pairs $s=(\sUV{X_i}{V'},
    \sUU{X_i})$ with $\sUV{X_i}{V'}$ being a list of linearly
    independent \\ \hspace{2cm} vectors from $\zo^{X_i}$ and $\sUU{X_i}$ a symmetric
    $X_i\times X_i$ matrix with entries from $\zo$ \\ \hspace{2cm} 
    -- cf.\ Definition~\ref{def:scenario}} \State $S(i, D, s) \gets 0$
  \EndFor \State $(j_1, j_2) \gets (\text{left child of $i$},
  \text{right child of $i$})$ \ForAll {scenarios $s_1, s_2$ for
    $|X_i|$ vertices}
\label{alg:join:loop}
    \State $s\gets \st{s_1}{s_2}{G\slt D [X_i]}$ \Comment{Definition~\ref{def:scenario_join}}
\label{alg:join:conv}
      \State $S(i,D,s) \gets S(i,D,s) + S(j_1, D, s_1)\cdot S(j_2, D, s_2)$
  \EndFor
  \EndProcedure
\end{algorithmic}
\end{algorithm}

\subsection{Introduce Nodes}

Introduce nodes are handled by Algorithm~\ref{alg:introduce}, which is
based on

\begin{lem}
\label{lem:sum_introduce}
Let $i$ be an introduce node with child $j$ and $X_i=X_j\uni \{a\}$.
Let $D\subs X_i$ and $s$ a scenario of $X_i$. Let
$D'=D\setminus\{a\}$. Then one of the following cases applies:
  \begin{itemize}
  \item If there is
    a scenario $s'$ of $X_j$ with $\sp{s'}{a}{G\slt D [X_i]} = s$, then
    we have $S(i,D,s)= S(j, D', s')$.
  \item Otherwise, $S(i, D, s)=0$.
  \end{itemize}
\end{lem}

\begin{proof}
  Assume there is some $(A,B)$ such that $\Sc(\EG{G'}{A\uni B}{X_i}) =
  s$. Let $s'=\Sc(\EG{G'}{A\uni B}{X_j})$. By
  Lemma~\ref{lem:rank_introduce} it follows $s=\sp{s'}{a}{G'[X_i]}$.
  Conversely, Lemma~\ref{lem:rank_introduce} also guarantees that for
  all $(A,B)$ with $\Sc(\EG{G'}{A\uni B}{X_j}) = s'$ and
  $\sp{s'}{a}{G'[X_i]}=s$ we have $\Sc(\EG{G'}{A\uni B}{X_i}) = s$.
 \end{proof}

\begin{algorithm}
\caption{Computing the parts at an introduce node.}
\label{alg:introduce}
\begin{algorithmic}[1]
  \Procedure{Introduce}{$i$, $D$}
  \ForAll {scenarios $s$ for $|X_i|$ vertices}
\label{alg:introduce:init}
    \State $S(i,D,s)\gets 0$
  \EndFor
  \State $j\gets \text{child of $i$}$
  \State $a \gets \text{vertex being introduced in $X_i$}$
  \ForAll {scenarios $s'$ for $|X_j|$ vertices}
\label{alg:introduce:loop}
    \State $s\gets \sp{s'}{a}{G\slt D[X_i]}$ \Comment{Definition~\ref{def:scenario_introduce}}
\label{alg:introduce:conv}
      \State $S(i, D, s) \gets  S(j, D\setminus \{a\}, s')$
    \EndFor
  \EndProcedure  
\end{algorithmic}
\end{algorithm}

\subsection{Forget Nodes}

Finally, let us consider Algorithm~\ref{alg:forget}, which handles
forget nodes. As $\dn {s'} a {G'}$ may be $-1$ in
Lines~\ref{alg:forget:deltan1} and \ref{alg:forget:deltan2}, we have
to assume $v\neq 0$. (The case $v=0$ is discussed in
Section~\ref{ssec:full_rank}.) Algorithm~\ref{alg:forget} is based on
\begin{lem}
\label{lem:sum_forget}
  Let $i$ be a forget node with child $j$ and $X_j=X_i \uni \{a\}$.
  Let $D\subs X_i$, $D'=D\uni \{a\}$ and $s$ a scenario of $X_i$.
  Then
  \begin{equation}
    \label{eq:sum_forget}
    \begin{split}
      S(i,D,s)= & \sum_{s'\in \S_\rmi} S(j, D, s') \\
                & + \sum_{s'\in \S_{\rmf}} x_a u^{\dr{s'}{a}{G\slt D[X_j]}} v^{\dn{s'}{a}{G\slt D[X_j]}} S(j, D, s') \\
                & + \sum_{s'\in \S_{\rmf'}} y_a u^{\dr{s'}{a}{G\slt D'[X_j]}} v^{\dn{s'}{a}{G\slt D'[X_j]}} S(j, D', s'),
    \end{split}
  \end{equation}
  where
  \begin{eqnarray*}
    \S_\rmi &=& \{ s'\ |\ \text{$s'$ scenario of $X_j$ with $\si{s'}{a}=s$}\}, \\
    \S_{\rmf} &=& \{ s'\ |\ \text{$s'$ scenario of $X_j$ with $\sm{s'}{a}{G\slt D [X_j]} = s$} \}, \\
    \S_{\rmf'} &=& \{ s'\ |\ \text{$s'$ scenario of $X_j$ with $\sm{s'}{a}{G\slt D'[X_j]}= s$} \}.
  \end{eqnarray*}
\end{lem}

\begin{proof}
  We use \eqref{eq:bubble_term} again. Let $(A,B)$ be
  admissible. There are three cases: (1) $a\not \in A\uni B$, (2)
  $a\in A$ and (3) $a\in B$. In case (1), the term corresponding to
  $(A,B)$ is contained in the first sum in \eqref{eq:sum_forget}. In
  case (2) we obtain the term corresponding to $(A,B)$ from the second
  sum in \eqref{eq:sum_forget}, where we use
  Lemma~\ref{lem:rank_forget} and multiply by $x_a$ to represent the
  fact that $a$ is in $A$. We also multiply by some power of $u$ and
  $v$ depending on the rank (nullity, resp.)  difference with
  vs.\ without $a$ in the extension.  Case (3) is similar, but we also
  have to use $D'$ instead of $D$ as in this case $a$ belongs to $B$
  and thus the self loop at $a$ is toggled.   \end{proof}

\begin{algorithm}
\caption{Computing the parts at a forget node.}
\label{alg:forget}
\begin{algorithmic}[1]
  \Procedure{Forget}{$i$, $D$}
  \ForAll {scenarios $s$ for $|X_i|$ vertices}
\label{alg:forget:init}
    \State $S(i, D, s) \gets 0$
  \EndFor
  \State $j\gets \text{child of $i$}$
  \State $a \gets \text{vertex being forgotten in $X_i$}$
  \ForAll {scenarios $s'$ for $|X_j|$ vertices}
\label{alg:forget:loop}
    \State $s \gets \si{s'}{a}$ \Comment{Definition~\ref{def:scenario_ignore}}
\label{alg:forget:conv0}
    \State $S(i,D,s) \gets S(i,D,s) + S(j,D,s')$
    \State $G'\gets G\slt D [X_j]$
    \State $s\gets \sm{s'}{a}{G'}$ \Comment{Definition~\ref{def:scenario_forget}}
\label{alg:forget:conv1}
    \State $S(i,D,s) \gets S(i,D,s)
        + x_a u^{\dr{s'}{a}{G'}} v^{\dn{s'}{a}{G'}} S(j,D,s')$
\label{alg:forget:deltan1}
    \State $D'\gets D \uni \{a\}$
    \State $G'\gets G\slt D'[X_j]$
    \State $s\gets\sm{s'}{a}{G'}$
\label{alg:forget:conv2}
    \State $S(i,D,s) \gets S(i,D,s)
        + y_a u^{\dr{s'}{a}{G'}} v^{\dn{s'}{a}{G'}} S(j,D',s')$
\label{alg:forget:deltan2}
  \EndFor
  \EndProcedure
\end{algorithmic}
\end{algorithm}

\subsection{Running Time}

We start with a nice tree decomposition with $O(n)$ nodes. Recall that
$k$ is the maximum bag size of the tree decomposition. To obtain the
vertex order (Algorithm~\ref{alg:order}) $O(n)\cdot \poly(k)$ steps
are sufficient.

The running time of Algorithm~\ref{alg:main} can be analyzed as
follows. The $i$ loop is executed $O(n)$ times, as there are $O(n)$
nodes in the tree decomposition. There are at most $2^k$ sets $D\subs
X_i$ for every node $i$.  There are at most $2^{(3k+1)k/2}$ scenarios
for $k$ vertices (Lemma~\ref{lem:scenario_bound}).  The join case
(Algorithm~\ref{alg:join}) sums over \emph{pairs} of scenarios and
thus dominates the running time of the introduce
(Algorithm~\ref{alg:introduce}) and forget
(Algorithm~\ref{alg:forget}) case. In the join case, we have to sum
over at most $(2^{(3k+1)k/2})^2$ pairs $(s_1, s_2)$. Converting the
scenarios (Line~\ref{alg:join:conv} of Algorithm~\ref{alg:join},
Line~\ref{alg:introduce:conv} of Algorithm~\ref{alg:introduce}, and
Lines~\ref{alg:forget:conv0}, \ref{alg:forget:conv1} and
\ref{alg:forget:conv2} of Algorithm~\ref{alg:forget}) takes time
polynomial in $k$, as we have shown in Section~\ref{sec:converting}.
Thus, the running time of Algorithm~\ref{alg:main} is at most
\[O(n) \cdot 2^k \cdot (2^{(3k+1)k/2})^2 \cdot \poly(k), \] if we
assume that arithmetic operations such as addition and multiplication
(of numbers) can be performed in one time step. The degree of the
interlace polynomial is at most $n$ in every variable (cf.\
Definition~\ref{def:interlace_polynom}). This leads to the following
result.
\begin{thm}
\label{thm:algo}
Let $G=(V,E)$ be a graph with $n$ vertices. Let a nice tree
decomposition of $G$ with $O(n)$ nodes and width $k$ be given, as well
as numbers $u$, $v$, $v\neq 0$, and, for each $a\in V$, $x_a$ and
$y_a$. Then Algorithm~\ref{alg:main} evaluates the multivariate
interlace polynomial $C(G)$ at $((x_a)_{a\in V},(y_a)_{a\in V},u,v)$
using $2^{3k^2+O(k)}\cdot n$ arithmetic operations. If the bit length
of $u, v$, and $x_a, y_a, a\in V$, is at most $\ell$, the operands
occurring during the computation are of bit length $O(\ell n)$.
\end{thm}

To evaluate the interlace polynomial of Arratia et al.\
\cite{arratia_two_var_interl}, which does not use self loop toggling
in its definition, we do not need parameter $D$ in
\eqref{eq:bubble_term} and the $D$-loop in Algorithm~\ref{alg:parts}.
This simplifies the algorithm a bit. The running time is also reduced,
but only by a factor $\leq 2^k$ and thus it is still $2^{3k^2+O(k)}n$.

If we consider path decompositions (see, for example, Bodlaender
\cite{Bodlaender19981}) instead of tree decompositions, we have no
join nodes. Thus, for graphs of bounded pathwidth, we get a result
similar to Theorem~\ref{thm:algo} but with running time reduced to
$2^{1.5k^2 + O(k)}\cdot n$.

\subsection{Full-Rank Induced Subgraphs -- The Case $v=0$.}
\label{ssec:full_rank}

If $v=0$, the summation in \eqref{eq:interlace_polynom} extends only
over the $A,B\subs V$, $A\cut B = \emptyset$, such that the adjacency
matrix of $G\slt B[A\uni B]$ has full rank. This sum can be evaluated
using essentially the same techniques we have developed so far. Let us
sketch briefly what changes have to be made.

Consider the situation described on
Page~\pageref{lab:scenario_motivation}, i.e.\ there is an extended
graph $\eG {V'} U$, and symmetric Gaussian elimination on $G$ using
$V'$ has been performed. The result is depicted in
Figure~\ref{fig:scenario_intuition01}. Let $\S$ denote the columns of
the $U\times V'$ submatrix that are not ``ruled'' by any $1$-entry of
the $V'\times V'$ submatrix. (These columns are indicated by question
marks in Figure~\ref{fig:scenario_intuition01}.) Then the following
holds: The adjacency matrix of $G[V'\uni U]$ has full rank only if
$\S$ is linearly independent. If $U=\emptyset$, the converse is also
true for trivial reasons. Following this observation, we can modify
our algorithm to count full-rank induced subgraphs only and thus
evaluate the interlace polynomial at points with $v=0$.

The first modification is to extend
Definition~\ref{def:extension_scenario} as follows: The scenario of an
extended graph $\EG G {V'} U$ is said to have \defexpr{full rank} if
the column set $\S$ defined as above is linearly independent.

Next, we replace \eqref{eq:bubble_term} by
\begin{equation}
  \label{eq:bubble_term'}
  S(i, D, s) = \sum_{A,B} x_A y_B u^{\rk((G_i\slt B)[A\uni B])},
\end{equation}
where the summation extends over all $A,B$ as in
\eqref{eq:bubble_term} with the additional restriction that the
scenario of $\EG {G'} {A\uni B} {X_i}$, $G'=G\slt (B\uni D)$, must
have full rank.

Following the arguments in Section~\ref{sec:converting}, it is
possible to prove that full-rank scenarios can be used with tree
decompositions in the same way as ordinary scenarios. For instance,
the following version of Lemma~\ref{lem:rank_join} handles the join of
full-rank scenarios:

\begin{lem}[Join for full-rank]
\label{lem:rank_join_full_rank}
Let $G=(V,E)$ be a graph, $U\subs V$, and $s_1, s_2$ be two scenarios
of $U$. Then exactly one of the following cases applies:
\begin{enumerate}
\item For all disjoint subgraphs $G[V_1]$ and $G[V_2]$ of $G$ such
  that
  \begin{enumerate}
  \item $G[V_1]$ and $G[V_2]$ may be extended by $U$ according to $G$,
  \item $\eG {V_1} U$ has full-rank scenario $s_1$, and
  \item $\eG {V_2} U$ has full-rank scenario $s_2$,
  \end{enumerate}
  the scenario of $\eG {V_1\uni V_2} U$ is $\st {s_1}{s_2} {G[U]}$ but
  it does not have full rank.
\item For the same family of graphs as in the first case, the
  following holds: The scenario of $\eG {V_1\uni V_2} U$ is $\st {s_1}
  {s_2} {G[U]}$ and it has full rank.
\end{enumerate}
Moreover, during the $\poly(|U|)$-time computation of $\st {s_1} {s_2}
{G[U]}$ as described in the proof of Lemma~\ref{lem:rank_join}, it can
be decided which of the two cases applies. We say that $\st {s_1}{s_2}
{G[U]}$ \defexpr{preserves full rank} if the second case applies.
\end{lem}


In the algorithm, scenario-sums must be counted only if the scenario
has full rank. For instance, join nodes can be handled by
Algorithm~\ref{alg:join_full_rank}, which is a slight modification of
Algorithm~\ref{alg:join}.

\begin{algorithm}
\caption{Computing the full-rank parts at a join node.}
\label{alg:join_full_rank}
\begin{algorithmic}[1]
  \Procedure{Join\_full\_rank}{$i$, $D$}
  \ForAll {scenarios $s$ for $|X_i|$ vertices}
  \State $S(i, D, s) \gets 0$
  \EndFor
  \State $(j_1, j_2) \gets (\text{left child of $i$}, \text{right child of $i$})$
  \ForAll {scenarios $s_1, s_2$ for $|X_i|$ vertices}
    \If {$\st{s_1}{s_2}{G\slt D [X_i]}$ preserves full rank}
      \State $s\gets \st{s_1}{s_2}{G\slt D [X_i]}$
      \State $S(i,D,s) \gets S(i,D,s) + S(j_1, D, s_1)\cdot S(j_2, D, s_2)$
    \EndIf
  \EndFor
  \EndProcedure
\end{algorithmic}
\end{algorithm}

In this way, Theorem~\ref{thm:algo} can be established for the case
$v=0$ as well.

\section{Variants of the Algorithm}
\label{sec:variants}

\subsection{Evaluation vs.\ Computation}
\label{ssec:eval_vs_comp}

The main motivation for our algorithm is \emph{evaluation} of the
multivariate interlace polynomial: We are given numerical values for
the variables $x_a, y_a, u, v$, an $n$-vertex graph $G$ and a nice
tree decomposition of $G$. From this, we want to compute the numerical
value $C(G;(x_a)_{a\in V}, (y_a)_{a\in V}, u, v)$. Our algorithm
solves this problem as described above.

Another problem one might be interested in is the \emph{computation}
of the interlace polynomial: Given $G$, output a description of the
polynomial $C(G)$, which is a polynomial over the indeterminates
$\{x_a, y_a \ |\ a\in V\} \uni \{u,v\}$. As the number of monomials of
$C(G)$ is exponential in $n$, there is no algorithm with running time
polynomial in $n$ that computes the multivariate interlace polynomial
if we represent $C(G)$ as a list of the coefficients of all the
monomials.  However, there are other ways of representing polynomials,
for example arithmetic formulas and arithmetic circuits, which are
considered in algebraic complexity theory \cite{Buerg1997}.

An \emph{arithmetic circuit} is a directed graph with nodes of
indegree 0 or 2. Nodes with indegree 0 are inputs and labeled by a
constant or a variable. They compute the polynomial they are labeled
with. Nodes with indegree two are labeled with plus or times and
compute the sum (product, resp.) of their children. We say that a
circuit computes a polynomial if it computes it at one of its nodes.

If one accepts arithmetic circuits as a compact way to describe
polynomials, then our algorithm actually \emph{computes} the
multivariate interlace polynomial: Use Algorithm~\ref{alg:main} as a
procedure to create an arithmetic circuit for the polynomial $C(G)$ in
the following way. Start with a circuit with inputs $x_a$ and $y_a$
for each $a\in V$, as well as inputs for $u$, $v$, $0$, and $1$. For
each operation of the algorithm of Section~\ref{sec:algo} using the
``parts'' $S(i,D,S)$, add gates that implement this operation. In this
way, the algorithm creates an arithmetic circuit $\CC$ of size
$2^{3k^2+O(k)}n$ that computes $C(G)$.

In the following two subsections, we use this point of view for
parallel evaluation and for computation of $d$-truncations of the
multivariate interlace polynomial.

\subsection{Parallelization}
\label{ssec:parallel}

In this subsection we discuss a way to parallelize our algorithm. We
do this using two operations on the tree decomposition: (1) removing
all leaves and (2) contracting every path with more than one node. Our
approach is not new but a variation of standard methods \cite[Section
2.6.1]{Leighton:1992}, \cite[Section 3.3]{jaja}.

To describe the operations, we need some formalism. We use vectors
$\sigma$ to collect the parts of the interlace polynomial which are
computed. For each node $i$ we define the vector $\sigma_i=(S(i,D,s)\
|\ D\subs X_i, \text{$s$ scenario of $X_i$})$, where the order of the
components of the vector is fixed appropriately.  We call $\sigma_i$
the ``output'' of node $i$. We call nodes with one child $1$-nodes and
nodes with two children $2$-nodes. Nodes without children are leaves.
Every $1$-node has one input vector $\sigma_j$ which is the output of
its child, every $2$-node has two input vectors which are the output
vectors of its children. By definition, for leaves the input and the
output is identical.

With each $1$-node $i$ with child $j$ we associate a matrix $A_i$. The
computation of the $1$-node $i$ is $\sigma_i = A_i \sigma_j$. For an
introduce node $i$ with child $j$, by Lemma~\ref{lem:sum_introduce} we
trivially can write $\sigma_i=A_{i} \sigma_j$ for some matrix $A_{i}$.
The entries of $A_i$ are either $0$ or $1$. Now let $i$ be a forget
node with child $j$. Consider \eqref{eq:sum_forget}. Note that in each
of the three sums, the question, which $S(j,D,s')$ ($S(j,D',s')$,
resp.) are used, i.\ e.\ over which $(D, s')$ ($(D',s')$, resp.) is
summed, can be answered considering only $G[X_j]$ and the involved
scenarios.  Thus, we can compute from this a matrix $A_{i}$ with
$\sigma_i = A_{i}\sigma_j$, too. The entries of $A_i$ are $0$, $1$,
$x_a u^l v^{1-l}$ or $y_a u^l v^{1-l}$, where $l\in \{0,1,2\}$.

Consider a $2$-node $i$ with children $j_1$ and $j_2$. The computation
performed at $i$ is
\begin{equation}
\label{eq:twonode}
\sigma_i(D,s)=\sum \sigma_{j_1}(D,s_1)\sigma_{j_2}(D,s_2),
\end{equation}
where the sum is taken over the same elements as in
\eqref{eq:sum_join}.

The parallel computation of the interlace polynomial works as follows.
We start with the nice tree decomposition of the input graph with
$O(n)$ nodes and an arithmetic circuit of constant depth which
computes $\sigma_i$ for all leaves $i$ of the tree decomposition and
$A_i$ for all matrices associated with any node $i$ of the tree
decomposition.  Then we reduce the tree underlying the tree
decomposition step by step. Every time we reduce the tree, we extend
the arithmetic circuit such that the above invariant is preserved.

We initialize the arithmetic circuit as follows: We insert the
constants $0$ and $1$, $u$, $v$ and for every vertex $a$ of $G$ we
insert $x_a$ and $y_a$. Then we produce all entries of all matrices
associated with any node of the tree decomposition in parallel. This
takes constant depth.

We repeat the following operations on the tree decomposition until it
consists only of one leaf: (1) contract all paths of $1$-nodes and (2)
remove all leaves.

\emph{Path contraction} works as follows. For a sequence $i_1, i_2,
\ldots, i_\ell$ of $1$-nodes we have $\sigma_{i_\ell} = A_{i_\ell}
\cdot \ldots \cdot A_{i_1}\sigma_j$, where $\sigma_j$ is the input of
node $i_1$. Thus, we can substitute the sequence by one node which has
$\AA=A_{i_\ell} \cdot \ldots \cdot A_{i_1}$ associated with it and
gets $\sigma_j$ as input. The depth of computing the matrix product in
parallel is $\Theta(\log \ell)$. Thus a step contracting any number of
disjoint $1$-nodes paths of length $\leq \ell$ increases the depth of
the arithmetic circuit by $\Theta(\log \ell)$.

Now we come to \emph{removal of leaves}. By this we mean the
following: Let $L$ be the set of all leaves of the tree decomposition.
Remove the elements of $L$ distinguishing the following cases: (1)
node $i$ has two children $j_1$ and $j_2$ which are both leaves, (2)
node $i$ has two children $j_1$ and $j_2$, one of which is a leaf
($j_1$, say) whereas the other is not, and (3) node $i$ has one child
$j$ which is a leaf. To handle case (1) we introduce a level with
multiplications and a level with additions to perform
\eqref{eq:twonode}. This increases the depth by $2$. In case (2) node
$i$ becomes a $1$-node: The $\sigma_{j_1}(D,s)$ in \eqref{eq:twonode}
become coefficients of a new matrix $\AA$ associated to $i$. As by the
invariant, the arithmetic circuit already computes the
$\sigma_{j_1}(D,s)$, we do not need any new gates and depth is not
increased. For case (3) we have to implement the matrix multiplication
$A_i\sigma_j$ to compute $\sigma_i$. This increases the depth by a
constant. Thus, removing all leaves in $L$ increases the depth only by
a constant.

After performing all possible path contractions, the number of
$1$-nodes is at most two times the number of $2$-nodes. Thus, at least
$1/4$ of the nodes are leaves. This implies that the following removal
of leaves decreases the number of nodes of the tree decomposition by a
factor of at least $1/4$. Thus, after $O(\log n)$ steps the tree
decomposition is reduced to a single leaf. In each step the depth
increases by at most $O(\log n)$, which gives a $O(\log^2 n)$ bound on
the depth of the constructed arithmetic circuit.

\subsection{Computation of the Coefficients}

\label{ssec:coefficients}

As discussed in Section~\ref{ssec:eval_vs_comp}, our algorithm can be
used to create an arithmetic circuit $\CC$ of size $2^{3k^2+O(k)}n$
that computes $C(G)$ for an $n$-vertex graph $G$ with appropriate tree
decomposition of width $k$.  Now one can apply standard techniques to
convert $\CC$ into a procedure computing some of the coefficients of
$C(G)$.

Let us elaborate this for an example, the computation of the
$d$-truncation of the multivariate interlace polynomial. Courcelle
defines the $d$-truncation \cite[Section 5]{courcelle_interlace_final}
of a multivariate polynomial as follows.  The \emph{quasi-degree} of a
monomial is the number of vertices that index its indeterminates. As
the $G$-indexed part of the monomials of the multivariate interlace
polynomial are multilinear, the quasi-degree of a monomial of $C(G)$
is the degree of its $G$-indexed part. For example, the quasi-degree
of the monomial $x_Ay_Bu^{r}v^{s}$ is $|A|+|B|$. The
\emph{$d$-truncation} $\trunc {P(G)} d$ of a polynomial $P(G)$ is the
sum of its monomials of quasi-degree at most $d$. Let $\M$ be a set of
monomials. If \[ f=\sum_{m\in \M} a_m m \] is a polynomial and
$\M'\subs \M$, we set
\[ \trunc f {\M'} = \sum_{m\in \M'} a_m m.\]

As we want to use a result on fast multivariate polynomial
multiplication which uses computation trees \cite[Section
4.4]{Buerg1997} as model of computation, we also formulate our result
in this model. In addition to the arithmetic operations (addition,
multiplication, division), also comparisons are allowed in this model.
Each of these operations is counted as one step.

\begin{thm}[{\cite[Theorem 1]{fastmpmul}}]
\label{thm:fastmpmul}
Consider polynomials over the indeterminates $x_1, \ldots, x_n$. Let
$d$ be a positive integer, and $\D$ the monomials of degree at most
$d$. Let $f, g$ be two polynomials.  Then, assuming the coefficients
of $\trunc f \D$ and $\trunc g \D$ are given, the coefficients of
$\trunc {(f\cdot g)} \D$ can be computed using
\[O(D(\log D)^3 \log(\log D))\] operations in the computation tree
model, where $D=|\D|$.
\end{thm}

\begin{cor}
\label{cor:coeffis}
Let $G$ be a graph with $n$ vertices. Let a nice tree decomposition of
$G$ with width $k$ and $O(n)$ nodes be given. Then the coefficients of
all monomials of the $d$-truncation of $C(G)$ can be computed using
\[ 2^{3k^2 + O(1)} n^{d(1+o(1))+O(1)}\] operations in the computation
tree model.
\end{cor}
Note that the $d$-truncation of $C(G)$ has more than $\binom n d \geq
n^{d(1-\log d / \log n)}$ monomials.

\begin{proof}[of Corollary~\ref{cor:coeffis}]
  Let us fix a $d$ and a graph $G$ with $n$ vertices and treewidth
  $k$. We want to compute the coefficients of the $d$-truncation of
  $C(G)$. As discussed in Section~\ref{ssec:eval_vs_comp}, there
  exists an arithmetic circuit $\CC$ of size $2^{k^3+O(k)}n$ computing
  $C(G)$. We convert every operation $f=g+h$ or $f = g \cdot h$ in
  $\CC$ into a sequence of operations computing the coefficients of
  each monomial of $\trunc f d$. In this way, we also get the
  coefficients of $\trunc{C(G)}d$.  To prove the corollary, it is
  sufficient to show that each operation is converted into at most
  $n^{d(1+o(1))+O(1)}$ operations.

  We start with additions. We convert every addition gate $f=g+h$ in
  $\CC$ into the operations $f_m=g_m+h_m$, $m\in \M$, where $\M$ is an
  appropriate set of monomials. The monomials of $\trunc {C(G)} d$ are
  a subset of $\M$ if $\M$ denotes the set of all monomials over
  $G$-indexed variables $x$ and $y$ and ordinary variables $u$ and $v$
  such that the quasi-degree is at most $d$ and the degree in $u$ and
  in $v$ is at most $n$. We can select a monomial in $\M$ in the
  following way.  First, choose $d$ times either $1$ or a variable
  from $\{x_a, y_a\ |\ a\in V\}$.  Then, choose the exponent of $u$
  and $v$ from $\{0, 1, \ldots, n\}$. Thus, we have
\begin{equation}
\label{eq:Mbound}
|\M|\leq(2n+1)^d (n+1)^2=n^{d\big(1+\frac {O(1)} {\log n}\big)+O(1)}.
\end{equation}
As we convert every addition from $\CC$ into $|\M|$ operations, the
claimed bound of the corollary is fulfilled.

Now let us consider multiplications, i.e.\ let $f=g\cdot h$ be a
multiplication gate in $\CC$. We use fast multivariate polynomial
multiplication for the $G$-indexed variables and the school method for
the ordinary variables. To this end, we fix the $u$- and $v$-part of
the monomial, i.e.\ we choose $d_u$ and $d_v$, $0\leq d_u, d_v\leq n$.
We want to compute the coefficients of the monomials $m$ of $f$ with
$\deg_u(m)=d_u$ and $\deg_v(m)=d_v$.  Choose nonnegative integers
$d_{u,g},d_{u,h}, d_{v,g}, d_{v,h}$ such that $d_{u,g}+d_{u,h}=d_u$
and $d_{v,g}+d_{v,h}=d_v$. Let
\[\D = \{x_A y_B\ |\ A, B\subs V(G), |A|+|B|\leq d\}.\] We can assume
that we have already computed all coefficients of $\tilde g :=\trunc g
{u^{d_{u,g}}v^{d_{v,g}}\D}$ and $\tilde h:=\trunc h {u^{d_{u,h}}
  v^{d_{v,h}} \D}$. (Here, an expression of the form $u^a v^b \D$
denotes the set $\{u^a v^b m\ |\ m\in \D\}$.) By
Theorem~\ref{thm:fastmpmul}, we can compute all coefficients of the
product $\tilde g \cdot \tilde h$ using
\[O(|\D|(\log |\D|)^3\log\log |\D|) = n^{d\big(1+\frac {O(1)}{\log
    n}\big)+\frac{O(\log\log n)}{\log n}}\] operations, as $|\D|\leq
(2n+1)^d\leq n^{d\big(1+\frac {\log 3} {\log n}\big)}$. We do this for
every choice of $d_{u,g}$, $d_{u,h}$, $d_{v,g}$, and $d_{v,h}$. As
these are at most $(n+1)^2$ many, this takes $n^{d\big(1+\frac {O(1)}
  {\log n}\big)+O(1)}$ steps. Adding the results monomial-wise needs
at most $|\D|(n+1)^2=n^{d\big(1+\frac {O(1)} {\log n}\big)+O(1)}$
additions and yields the coefficients of $\trunc f
{u^{d_u}v^{d_v}\D}$. We do this for all $(n+1)^2$ choices of $d_u$ and
$d_v$ to obtain the coefficients of all monomials of the
$d$-truncation of $f$. Thus, each multiplication in $\CC$ is converted
into $n^{d\big(1+\frac{O(1)} {\log n}\big) + O(1)}$ operations. This,
again, is within the claimed bound of the corollary.
 \end{proof}

\section{Further Questions}

\label{sec:open}

If we consider graphs of bounded cliquewidth instead of treewidth, so
called $k$-expressions take the role of tree decompositions. Our
concept of scenarios is tailor-made for tree decompositions and does
not work with $k$-expressions. Is there a linear algebra approach,
possibly similar to the one we presented in this work, to compute the
interlace polynomial using $k$-expressions?

The notion of rankwidth, which is related to cliquewidth
\cite{DBLP:journals/jct/Oum05, DBLP:journals/jct/OumS06}, is defined
using the $GF(2)$-rank of some matrices derived from a graph.
Furthermore, local complementation is studied in the context of the
interlace polynomial as well as in the context of rankwidth
\cite[Section 2]{DBLP:journals/jct/Oum05}. Thus, it seems to be
possible that rank decompositions support the computation of the
interlace polynomial very nicely. We have not investigated this
question in detail and leave it as a direction for further research.

\section*{Acknowledgements}
We would like to thank Bruno Courcelle and the anonymous referees for
their helpful comments.

\bibliographystyle{alpha}
\bibliography{literatur}

\end{document}